\documentclass[reqno,12pt]{amsart}

\usepackage{amsthm,amsmath,amsfonts,amssymb}

\usepackage{graphicx}  % to include pdf pictures, use PDFTeXify instead of usual TeXify that produces Dvi
\usepackage{subfigure}
\usepackage{caption}
\usepackage{bm}
\usepackage[export]{adjustbox}
\usepackage{subcaption}
\usepackage{multirow}
\usepackage{wrapfig}
\usepackage{floatflt}
\usepackage{ragged2e}
\usepackage{placeins}
\usepackage{lipsum}
\usepackage{float}
\usepackage[font={small}, singlelinecheck=false]{caption}

\numberwithin{equation}{section}

\theoremstyle{plain}
\newtheorem{thm}{Theorem}[section]
\newtheorem{exm}{Example}[section]
\newtheorem{prop}{Proposition}[section]
\newtheorem{cor}{Corollary}[section]
\newtheorem{rem}{Remark}[section]

\newtheorem{dfn}{Definition}[section]

\def\wc{\overset{d}{=}}
\def\wcL{\overset{{\mathcal{L}}}{=}}

\def\vr{\upsilon}
\def\bE{\mathbb{E}}
\def\bR{\mathbb{R}}
\def\bP{\mathbb{P}}
\def\ub{\mathbf{w}}
\def\tw{\tau_{{}_{{\mathbf{w}}}}}
\def\SR{\mathrm{SR}}
\def\Ri{\mathrm{Ri}}

\newcommand\bld[1]{\boldsymbol{#1}}

\begin{document}
\today

\title{Optimal Betting: Beyond the Long-Term Growth}
\author{Levon Hakobyan}
\thanks{Department of Mathematics, University of Southern California, Los Angeles, CA 90089,
 levon@usc.edu, lototsky@usc.edu  {\tt https://dornsife.usc.edu/sergey-lototsky/}}
%\thanks{}

 %

\author{Sergey Lototsky}

%

% REQUIRED
\begin{abstract}
While the Kelly portfolio has many desirable properties, including optimal long-term growth rate, the resulting investment strategy is rather aggressive. In this paper, we suggest a unified approach to the risk assessment of the Kelly criterion in both discrete and continuous time by introducing and analyzing  the asymptotic variance that describes  fluctuations of the portfolio growth, and use the results to  propose two new measures for quantifying  risk.
\end{abstract}

% REQUIRED


\maketitle	

\section{Introduction}
Investing  is a complex process that includes multiple steps. One of the main challenges in this process is capital allocation. In 1956, John Kelly \cite{Kelly56} introduced a simplified version of the capital allocation problem by considering a gambler who is placing sequential bets on a favorably biased coin. The gambler starts with some initial  wealth and knows the odds. The objective is to devise a strategy that outperforms all alternatives according to a specified performance metric. Kelly proposed the metric  based on maximizing the long-term growth rate of wealth;  the corresponding strategy, known as the Kelly criterion/strategy, consists in betting a fixed fraction $f^*$ of the current overall wealth. Ed Thorp \cite{Thorp03,Thorp06,ThorpPortfolio,Thorp08} implemented and tested the strategy in a variety of settings, including stock market.
% Maybe write more about how this setup was then Generalized
%Are these general conditions or do we need to generalize it further? discrete time- iid returns and continuous time- lognormal distribution

The Kelly criterion has an alternative interpretation in terms of  classical utility theory, when the objective is to maximize the expected value of a certain function of the terminal wealth.
Under fairly general conditions on the market, such as iid returns in discrete time or lognormal distribution of returns in continuous time, maximizing the  long-term growth rate is equivalent to maximizing the expected value of the logarithm of the terminal wealth.  Logarithmic utility is a member of a larger class of functions called isoelastic, or power, utilities indexed by  risk aversion parameter. In continuous time, Merton's Fund Separation Theorem \cite{Merton-Fund-Separation-69} states that, for every power utility function, the optimal solution is a fixed multiple of the  Kelly strategy $f^*$.  In practical terms, the strategy means that  the investor may over-bet by using a fraction bigger than $f^*$,   or under-bet by using a smaller fraction. According to Bellman and Kalaba \cite{Bellman-Kalaba}, no other reasonable utility function leads to a constant  optimal strategy. Note that, in the long run, over-betting can lead to ruin and under-betting reduces the growth rate.
% (need to understand the assumptions on the returns in both discrete and continuous time). Also need a coherent transition to fractional Kelly as a motivation for what we are doing.

Kelly strategy has low risk aversion: the fraction $f^*$  is relatively high and fluctuations of the  wealth are large.  One way to reduce these fluctuations is by using fractional Kelly strategy:
betting a   proportion $f<f^*$. There are various approaches to selecting the fraction. MacLean, Ziemba and Blazenko \cite{MacLean-growth-security-92}  investigate the problem  in discrete time. They consider different criteria based on several combinations of growth and security measures, such as  expected log terminal wealth and the  probability that the wealth stays above a specified path, and solve the corresponding constrained optimization problems. Each solution leads to an efficient frontier similar to the one in the Markowitz Modern Portfolio Theory \cite{Markowitz-52}. Some of the strategies on the efficient frontier are dynamic rather then  fixed fractions; cf.  Gottlieb \cite{Gottlieb-85}.  Still,  certain  convex combinations of optimal growth and optimal security  lead to traditional fractional Kelly criterion. MacLean and Ziemba \cite{MacLean-growth-security-99}  carry out the same analysis in continuous time with log-normal distribution of asset prices.

The fractional Kelly criterion has a statistical interpretation. Any practical implementation of the  Kelly strategy requires estimation of the mean and covariance matrix of returns, and  estimation errors can lead to over-betting and subsequent ruin. Rising and Wyner \cite{Rising-Wyner-mean-shrinkage-2012}  address this problem  by using shrinkage estimator of mean return toward the lower risk-free rate. Their results suggest that a fractional Kelly strategy with full information is
equivalent  to a full Kelly strategy based on shrinkage estimators of the market returns. Building on this idea, Han, Yu and Mathew \cite{Han-Yu-Mathew-shrinakge-2019} show that if the log returns are jointly Gaussian, then the plug-in estimator of the Kelly portfolio is biased and
 overestimates the true Kelly portfolio weights. They suggest an unbiased Kelly estimator of the
 portfolio weights and an estimator that minimizes average loss of growth rate, and demonstrate
 that both estimators lead to fractional Kelly strategies.

The Kelly criterion follows from the law of large numbers (LLN) applied to the logarithm of the wealth process.  The fluctuations around the limit are often described by the central limit theorem (CLT); we refer to these fluctuations as the second-order correction.
In this paper, we show that every fractional Kelly strategy can be realized using a simple CLT-based risk measure. More precisely, we introduce the asymptotic variance of the  long-term growth rate and conduct an in-depth analysis of this quantity under different market conditions, both in discrete and continuous time. We show that asymptotic variance is an increasing function of the fraction $f$ and can be used to penalize risky investments. We introduce new asymptotic quantities, Sharpe Ratio and  ridge coefficient, and use them to find optimal trading strategies that take into account the investment risk through the asymptotic variance of the wealth process. The result is a systematic approach to the construction of  fractional Kelly criteria and is a  computationally efficient alternative  to the security measures used in \cite{MacLean-growth-security-92}.

Section \ref{sec:DT} introduces the new measures of risk in the context of  discrete time models. Section \ref{sec:HFC} outlines a systematic approach to high-frequency compounding. Section \ref{sec:CT} is about the continuous-time models, when the return process is a continuous semi-martingale. In such models, many quantities of interest, including the new measures of risk, often admit simple closed-form expressions.

Throughout the paper, $\bE$ denotes the expectation, $\mathrm{Var}$ denotes the variance,
and $\wc$ denoted equality in distribution. Occasionally, we use $\wcL$ to denote equality in law for random processes. The Greek letter $\upsilon$ is used to denote the asymptotic variance; this letter looks slightly different from the Latin $v$.

\section{Discrete Time}
\label{sec:DT}

\subsection{The Standard Model}
In this section, we review the derivation of the Kelly criterion using the Law of Large Numbers and introduce the second-order correction using the Central Limit Theorem. The limiting variance will be the key to quantifying the risk and deriving various versions of the fractional Kelly criterion.

\begin{dfn}
Given a sufficiently rich probability space $(\Omega,\mathcal{F},\bP)$, the {\tt standard model} of wealth growth is  the sequence $W_n^f$ defined by
\begin{equation}
\label{wealth2}
W_n^f = W_0 \prod_{k=1}^n\big(1 + f r_k\big), \ n=1,2,\ldots,
\end{equation}
where
\begin{enumerate}
\item $W_0$ is the initial wealth, and we usually set $W_0=1$;
\item the random variables $r_k,\ k\geq 1,$ representing returns on  each bet, are independent  and have the
 same distribution as a given random variable $r$;
\item the random variable $r$ satisfies
\begin{align}
\label{r1}
&\mathbb{P}(r\geq -1)=1;\\
\label{r2}
&\mathbb{P}(r>0)>0,\ \mathbb{P}(r<0)>0;\\
\label{r3-2}
&\mathbb{E}\big[\ln(1+r)\big]^2<\infty;
\end{align}
\item   the {\tt strategy}, that is, the proportion $f$ of the overall wealth used on each step, is constant and satisfies the {\tt no-short-no-leverage}  (NS-NL) condition, namely,
    $f\in [0,1]$;  the remaining $(1-f)$ fraction of the wealth is kept aside, with no interest earned.
\end{enumerate}

\end{dfn}

Condition \eqref{r1} quantifies the idea that a loss in a bet should not be
more than $100\%$. Condition \eqref{r2} is basic non-degeneracy: both gains and losses are possible.
Condition \eqref{r3-2} is a minimal requirement to define the
 long-term growth rate of the wealth process and to quantify the resulting risk. The NS-NL condition ensures that simple compounding \eqref{wealth2}, without
 margin calls or borrowing costs, is a reasonably realistic model.
 The second part of this section discusses some of the  extensions of the standard model.

The key objects in this section will be the functions
\begin{equation}
\label{F(f)}
g_r(f)=\mathbb{E}\big[\ln(1+fr)\big],\ \vr_r(f)=\mathrm{Var}\big[\ln(1+fr)\big]=\mathbb{E}\left[\ln(1+fr)\right]^2-g_r^2(f).
\end{equation}

By the law of large numbers (LLN) and the central limit theorem (CLT),
\begin{align}
\label{rate}
\lim_{n\to \infty} \frac{\ln W^f_n}{n}=\lim_{n\to \infty}\frac{1}{n}\sum_{k=1}^n\ln(1+fr_k)=g_r(f) \ \text{with probability one},\\
\label{eq:LLN-g2}
\lim_{n\to \infty}\frac{1}{\sqrt{n}}\sum_{k=1}^n\Big(\ln(1+fr_k)-g_r(f) \Big)=\mathcal{N}\big(0,\vr_r(f)\big) \ \text{in distribution},
\end{align}
where $\mathcal{N}(\mu,\sigma^2)$ is a Gaussian random variable with mean $\mu$ and variance $\sigma^2$.
In particular,
\begin{equation}
\label{W-2order}
W^f_n=\exp\Big( n g_r(f)+\sqrt{n}\big(\vr_r(f)\zeta_n+
\epsilon_n\big)\Big),
\end{equation}
where  $\zeta_n,\ n\geq 1, $ are standard Gaussian random variables and
\begin{equation*}
%\label{asymp2}
\lim_{n\to \infty}\epsilon_n=0
\end{equation*}
in probability; see   \cite[Theorem 2.2]{L-P} for details.

\begin{prop}
\label{prop:mean-var}
Under assumptions \eqref{r1}--\eqref{r3-2}, the function $f\mapsto g_r(f)$ is concave and the function $f\mapsto \vr_r(f)$ is increasing, $f\in (0,1)$.
\end{prop}

\begin{proof}
By direct computation,
\begin{align*}
\frac{d^2g_r}{df^2}(f)&=-\mathbb{E}\left[\frac{r^2}{(1+fr)^2}\right]<0;\\
\frac{1}{2}\frac{d\vr_r}{df}(f)&=\bE\left[ \frac{r\ln(1+fr)}{1+fr}\right]- \bE\big[\ln(1+fr)\big]  \bE\left[\frac{r}{1+fr}\right] >0.
\end{align*}
The second inequality follows after noticing that, for fixed $f\in (0,1)$, the functions $F(x)= \ln(1+fx)$ and $G(x)= x/(1+fx)$ are
monotonically increasing on $(-1/f,+\infty)$; then, taking $\tilde{r}$ an iid copy of $r$, we get
$$
0<\bE\Big[ \big(F(r)-F(\tilde{r})\big)\big(G(r)-G(\tilde{r})\big)\Big]=2\bE\big[F(r)G(r)\big]-2\bE \big[F(r)\big]\bE \big[G(r)\big].
$$
\end{proof}

Note that \eqref{r3-2} implies $g_r(0+)=\vr_r(0+)=0$. Moreover, if $\bE|r|<\infty$, then
\begin{equation}
\label{derivatives1}
g'_r(0+)=\bE[r],
\end{equation}
and if $\bE[r^2]<\infty$, then
\begin{equation}
\label{derivatives2}
g''_r(0+)=-\bE[r^2],\ \vr''_r(0+)=\mathrm{Var}[r].
\end{equation}
Unlike $g_r$, the functions $f\mapsto \vr_r(f)$ and $f\mapsto \sqrt{\vr_r(f)}$ can have inflection points in $(0,1)$, depending on the distribution of $r$ (see \eqref{Ch0} below).

Assumptions \eqref{r1}--\eqref{r3-2} do not guarantee that the wealth process \eqref{wealth2} is growing with $n$. Accordingly,
each of the following additional conditions ensures asymptotic   growth, and can be used as a definition of an {\tt edge} or a {\tt favorable game}:
\begin{align}
\label{Er} \bE [r] &>0,\\
\label{Elnr} g_r(f)=\bE [\ln(1+fr)]&>0 \ \ \text{for some}\  f>0,\\
\label{Eprime} g'_r(0+)=\lim_{f\to 0+} \mathbb{E}\left[\frac{r}{1+fr}\right]&>0.
\end{align}
Because, for $x>-1$,  $x\not=0$,
\begin{equation}
\label{basic}
\frac{x}{1+x}<\ln(1+x)<x,
\end{equation}
we have \eqref{Er} $\Rightarrow$ \eqref{Elnr} $\Rightarrow$ \eqref{Eprime}. To confirm \eqref{basic} for $x>0$, note that
$(1+t)^{-2}<(1+t)^{-1}<1,\ t>0,$ and integrate from $0$ to $x$. Also, if \eqref{r3-2} holds, then  \eqref{Elnr} and  \eqref{Eprime} are equivalent.

If \eqref{Eprime} holds, then concavity of $g_r$ (Proposition \ref{prop:mean-var}) implies existence of a unique $f^*>0$
such that
\begin{equation}
\label{f*}
f^*=\underset{f}{{\mathrm{argmax}}}\, g_r(f),\ g'_r(f^*)=0,
\end{equation}
 ensuring the largest possible asymptotic growth rate of
$W^f_n$; if, in addition, $g_f'(1-)<0$, then also $f^*<1$; see \cite[Proposition 2.5]{L-P} for details. In this setting, the {\tt Kelly criterion/strategy}
 refers to the choice $f=f^*$ in \eqref{wealth2}, and {\tt fractional Kelly criterion/strategy}
 means using $f\in (0,f^*).$

For an alternative interpretation of the Kelly criterion, define
\begin{equation}
\label{TtoGoal}
\tw(f)=\inf\left\{ n>0: W^f_n>\ub\right\},\ \ \ub>1.
\end{equation}
By \eqref{wealth2},
$$
\tw(f)=\inf\left\{ n>0: \sum_{k=1}^n \ln (1+fr_k) >  \ln \ub\right\};
$$
if  $g_r(f)>0$, then many properties of the $\tw(f)$ follow from renewal theory for random walks with positive drift \cite[Chapter 3]{Gut-RW}. In particular,
\begin{align}
\label{TtoG-A}
&\lim_{\ub\to \infty} \frac{\tw(f)}{\ln \ub }= \frac{1}{g_r(f)} \ \text{with probability one: \cite[Theorem 3.4.1]{Gut-RW}};\\
\label{TtoG-CLT}
&\lim_{\ub\to \infty} \sqrt{\frac{g_r(f)}{\ln\ub}}\left(\tw(f)-\frac{\ln\ub}{g_r(f)}\right)
=\mathcal{N}\left(0, \frac{\vr_r(f)}{g_r^2(f)}\right)\
 \text{in distribution: \cite[Theorem 3.5.1]{Gut-RW}}.
\end{align}

As a result, the  strategy $f^*$ maximizing the asymptotic rate of growth \eqref{rate} also minimizes the asymptotic   time \eqref{TtoG-A} to reach a goal;
see \cite[Section 3]{Kelly-Breiman} for more details.
 On the one hand, the conclusion is not
surprising. On the other hand, mathematically speaking, \eqref{rate} is more general than  \eqref{TtoG-A}:
 for example, \eqref{rate} does not require
$g_r(f)>0$.

To summarize, under conditions \eqref{r1}--\eqref{r3-2} and \eqref{Elnr}, there exists a unique $f^*$ such that $g_r'(f^*)=0$ and
\begin{equation}
\label{f-opt}
f^*=\underset{f}{{\mathrm{argmax}}}\, \lim_{n\to \infty}\frac{\ln W^f_n}{n} = \underset{f}{{\mathrm{argmax}}}\,\lim_{\ub\to \infty} \frac{\tw(f)}{\ln\ub}.
\end{equation}
The strategy $f^*$ has other properties. For example, the sequence $W^{f_n}_n/W^{f^*}_n, n\geq 1,$ is a supermartingale
for every {\em non-anticipating}  strategy   $f_n, \ n\geq 1$ \cite[Theorem 1]{Kelly-Finkel}, which somewhat justifies the initial  assumption of constant $f$ in \eqref{wealth2}.
Still, in terms of asymptotic relation \eqref{W-2order}, most of existing results concentrate on the LLN term \eqref{rate},
without taking the CLT correction  \eqref{eq:LLN-g2} into account.

To understand the effects of the central limit theorems \eqref{eq:LLN-g2} and \eqref{TtoG-CLT} on the asymptotic behavior of  the wealth process \eqref{wealth2},
define the {\tt asymptotic Sharpe ratio} of $W_n^f$ by
\begin{equation}
\label{SR-a}
\SR_r(f)=\frac{g_r(f)}{\sqrt{\vr_r(f)}}.
\end{equation}

Both \eqref{eq:LLN-g2} and \eqref{TtoG-CLT} demonstrate how   $\SR_r(f)$ can quantify the risk associate with \eqref{wealth2}.
\begin{enumerate}
\item By \eqref{eq:LLN-g2}, the probability that, for large $n$, $W_n^f$ drops below  a fixed threshold $\bld{w}$ is approximately
$\displaystyle \Phi\left(\frac{\ln \bld{w}-n\,g_r(f)}{\sqrt{n\,\vr_r(f)}}\right)$, with $\Phi$ denoting the standard normal cdf; this can be further approximated by $\Phi\big(-\sqrt{n}\,\SR_r(f)\big)$, which, for $g_r(f)>0$, is a decreasing function of $ \SR_r(f)$.
\item Similarly, \eqref{TtoG-CLT} suggests that large values of $\SR_r(f)$ reduce  fluctuations of  $\tw(f)$.
\end{enumerate}

Here are some  {\em local} properties of the asymptotic Sharpe ratio.
 \begin{thm}
\label{th-Sharpe}
Under conditions \eqref{r1}--\eqref{r3-2} and \eqref{Elnr}, the function  $f\mapsto \SR_r(f)$ is monotonically decreasing near $f^*$. Moreover, if
$\bE[r^2]<\infty$, then
\begin{equation}
\label{SR-r}
\SR_r(0+)=\frac{\bE[r]}{\sqrt{\mathrm{Var}[r]}},
\end{equation}
the Sharpe ratio of $r$.
\end{thm}

\begin{proof}
By direct computation,
$$
\SR_r'(f)=\frac{d}{df}\big(\SR_r(f)\big)  = \frac{2g'_r(f)\vr(f)-g_r(f)\vr'(f)}{2(\vr(f))^{3/2}},
$$
and, by Proposition \ref{prop:mean-var}, $\SR_r'(f^*)<0$.

If $\bE[r^2]<\infty$, then \eqref{derivatives1} and \eqref{derivatives2} imply
$$
\SR_r(f)\approx \frac{f\, \bE[r]}{f\,\sqrt{\mathrm{Var}[r]}}, \ \ f\to 0+,
$$
which is \eqref{SR-r}.
\end{proof}

An alternative way to use $\vr_r(f)$ for risk control is via the
{\tt asymptotic ridge coefficient} of $W_n^f$:
\begin{equation}
    \Ri_r(f,\gamma)=g_r(f) - \gamma\vr_r(f), \ \  \gamma\geq0,
\end{equation}
where  $\gamma$ can be interpreted as a risk-aversion parameter. The term {\em ridge coefficient}
is motivated by a similar construction in regression analysis \cite[Section 3.4.1]{hastie-09},
when quadratic penalization pushes the optimal solution closer to zero. The next result shows that,
for $\gamma>0$,  maximizing $\Ri_r$ with respect to $f$ leads to a fractional Kelly strategy.

\begin{prop}\label{prop-Ri-loc-max}
    Under conditions \eqref{r1}--\eqref{r3-2} and \eqref{Eprime}, for every
    $\gamma\in(0,+\infty)$, the function $f\mapsto \Ri_r(f,\gamma)$ achieves its global maximum in $[0,1]$ at a point $f^{\Ri}\in [0,f^*)$.
\end{prop}
\begin{proof}
Fix $\gamma > 0$. The proof of Proposition \ref{prop:mean-var} shows that, under assumption \eqref{Eprime},
$$
\frac{d}{df}\Ri_r(0+,\gamma)>0 \ \ {\text {and} } \ \ \frac{d}{df}\Ri_r(f,\gamma)<0,\ \
 f\geq f^*.
$$
The result now follows from the intermediate value theorem.
\end{proof}

In what follows, we refer to $f^{\Ri}$ as the {\tt optimal ridge strategy.}
% \begin{prop}
%     Under conditions \eqref{r1}--\eqref{r3-2} and \eqref{Er}, for any $\gamma\in[0,+\infty)$, the function $f\mapsto \Ri_r(f,\gamma)$ achieves its unique maximum at $f^\Ri(\gamma) = f^{\Ri}\in (0,f^*)$.
% \end{prop}
% \begin{proof}
%     Fix $\gamma\geq0$. Since the function $g_r(f)$ attains its maximum at $f^*$ and $\frac{d}{df}\vr_r(f)>0$ everywhere, including at $f^*$, it follows that $\frac{d}{df}\Ri_r(f^*,\gamma)<0$. On the other hand, $\frac{d}{df}\Ri_r(0+,\gamma) = \bE[r]>0$. By the intermediate value theorem, this implies the existence of $f^{\Ri}\in (0,f^*)$ such that $\Ri_r(f^{Ri},\gamma)>\Ri_r(f,\gamma)$, for any $f\in[0,1]$.
% \end{proof}
% Since the function g_r(f) is decreasing (f^*,1] and \vr_r(f) is increasing in general this must be unique maximum in [0,1]

%When there is no danger of confusion, we will write $\Ri_r(f)$ instead of
%$\Ri_r(f,\gamma)$.

Instead of penalizing high values of the asymptotic variance,
one can  impose an upper bound on $\vr_r(f)$: given $v_0\in (0,\vr_r(1))$,
\begin{equation}
\label{eq-constrained-optimization}
     \text{maximize } g_r(f), \text{ {subject to the constraint} } \vr_r(f)=  v_0;
\end{equation}
properties of the functions $g_r$ and $\vr_r$ imply that
 the equality constraint in \eqref{eq-constrained-optimization} is equivalent to $\vr_r(f)\leq   v_0$.
 Of course,  the solution to  \eqref{eq-constrained-optimization} is $f_0\in (0,f^*)$ such that $v_0 = \vr_r(f_0)$, implying that every fractional Kelly strategy can be realized as
a solution of \eqref{eq-constrained-optimization}.

On the other hand, the asymptotic ridge coefficient emerges as the Lagrangian for
\eqref{eq-constrained-optimization}, with $\gamma$ playing the role of the Lagrange multiplier,
leading to the system of equations
\begin{equation*}
    \frac{\partial\, \Ri_r(f,\gamma)}{\partial f}= g_r'(f) - \gamma \vr_r'(f) = 0, \quad \vr_r(f) = v_0 = \vr_r(f_0),
\end{equation*}
with solution
% $f=f_0$ and $\gamma=g_r'(f_0)/\vr_r'(f_0)$
\begin{equation*}
    f=f_0 \quad \text{and} \quad \gamma=\frac{g_r'(f_0)}{\vr_r'(f_0)}.
\end{equation*}

%This means that if $f_0\in[0,f^*)$ and $\gamma = \gamma^0 = g_r'(f_0)/\vr_r'(f_0)$, then $f_0$ becomes %a critical value for $\Ri_r(f,\gamma^0)$. This does not guarantee that $f_0$ is a local or global %maximum point, even though Proposition \ref{prop-Ri-loc-max} ensures the existence of a global %maximum in $(0,f^*)$. Nevertheless, this suggests that using the asymptotic ridge coefficient as a %criterion may naturally lead to a fractional Kelly strategy. Moreover, there are very special cases, %such as when $\vr_r(f)$ is a convex function, $\Ri_r(f,\gamma)$ serves as a universal criterion %covering all possible fractional Kelly strategies.

Unlike $g_r$, $\vr_r$ and $\Ri_r$, the {\em global} behavior of function $f\mapsto \SR_r(f)$ is more sensitive to the particular distribution of $r$; the following two examples
illustrate this sensitivity.

\begin{exm}
\label{exm-BM}
{\rm Consider the {\tt simple Bernoulli model}, with
$$
\bP(r=1)=p>1/2,\ \ \bP(r=-1)=1-p.
$$
Then
$$
g_r(f)=p\ln(1+f)+(1-p)\ln(1-f),\ f^*=\underset{f}{{\mathrm{argmax}}}\;g_r(f)=2p-1,
$$
\begin{equation}
\label{var-bern}
\vr_r(f)=p(1-p)\big(\ln(1+f)-\ln(1-f)\big)^2,
\end{equation}
and
\begin{equation}
\label{SR-bern}
\SR_r(f)=\frac{p\ln(1+f)+(1-p)\ln(1-f)}{\sqrt{p(1-p)}\big(\ln(1+f)-\ln(1-f)\big)}.
\end{equation}
In particular,  $\SR_r'(f)<0,\ f\in (0,1)$, whereas the functions $f\mapsto {\vr_r(f)}$
and $f\mapsto \sqrt{\vr_r(f)}$ are increasing and convex on $(0,1)$.
Figure 1(a) presents the  graphs of $f\mapsto g_r(f)$ and $f\mapsto \SR_r(f)$ in the simple Bernoulli model when $p=0.75$.
%Note that $\SR_r(0+)\approx 0.58$, which is close to $0.5/(\sqrt{3}/2)$, the Sharpe ratio for $r$,
% and is consistent with \eqref{SR-r}.

In this case,
$$
f^{*} =2p-1=0.5,\ \SR_r(f^*)\approx 0.27,\ g_r(f^*)=0.13,
$$
and any reduction of $f$ will increase $\SR_r$ as the optimal strategy according to Sharpe ratio measure is $f=0$. For example,  changing $f$ from  $f^*$ to $f=0.25$  leads to about $60\%$
increase of $\SR_r$ and a $30\%$ decrease of $g_r$: $\SR_r(0.25)\approx 0.43$, $g_r(0.25)\approx 0.10.$

At the same time

\begin{equation*}
    \Ri_r(f,\gamma) = p\ln(1+f)+(1-p)\ln(1-f) - \gamma{p(1-p)}\big(\ln(1+f)-\ln(1-f)\big)^2.
\end{equation*}

Optimizing this function reduces to solving the following equation:

\begin{equation*}
    4\gamma p(1-p)\ln\bigg(\frac{1+f}{1-f}\bigg)+f+1-2p = 0.
\end{equation*}

When \( p = 0.75 \) and \( \gamma = 1 \), we obtain
\begin{equation}
    f^{\Ri} \approx 0.2,
\end{equation}

which is nearly 30\% lower than the fraction suggested by Kelly criterion.

Examining the results, we see:
\begin{equation}
    g_r(f^{\Ri}) \approx 0.08 \quad \text{versus} \quad g_r(f^{*}) \approx 0.13.
\end{equation}

Additionally,
\begin{equation}
    \vr_r(f^{\Ri}) \approx 0.03 \quad \text{compared to} \quad \vr_r(f^{*}) \approx 0.22.
\end{equation}

In other words, when following the optimal ridge strategy,
we sacrifice about $30\%$ of growth in exchange for a nearly $90\%$ reduction in variance.}

This concludes Example \ref{exm-BM}.
\end{exm}

\begin{figure}
\subfigure[Simple Bernoulli model]{
\includegraphics[width=0.46\linewidth]{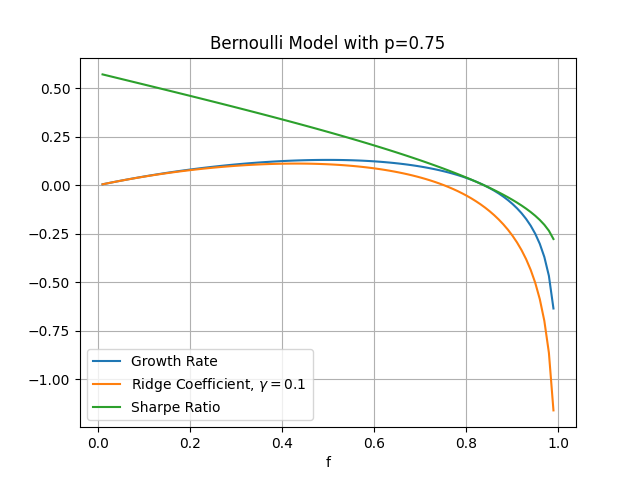}}
\subfigure[Squared Cauchy model]{
\includegraphics[width=0.46
\linewidth]{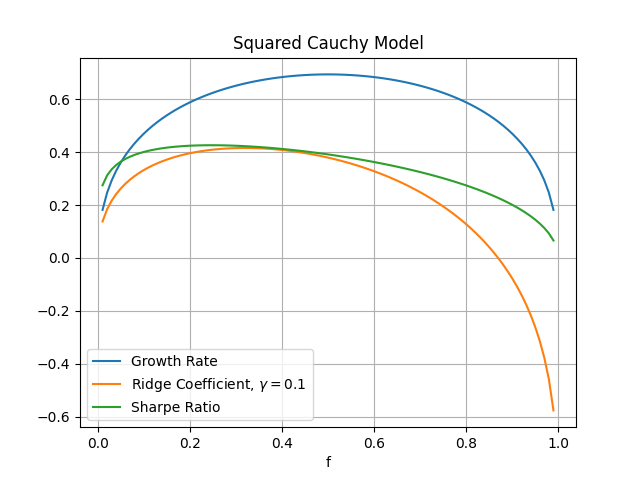}
}
\caption{Asymptotic Growth Rate and Sharpe Ratio For Two Different Models}
\label{fig:res_u_v}
\end{figure}
\FloatBarrier

\begin{exm}
\label{exm-CM}
{\rm Consider
\begin{equation}
\label{Ch0}
r=\eta^2-1,
\end{equation}
where $\eta$ has standard Cauchy distribution with probability density function
$$
h_{\eta}(x)=\frac{1}{\pi(1+x^2)},\  \  \ -\infty< x<+\infty.
$$
Then
$$
g_r(f)=\frac{2}{\pi}\int_0^{+\infty} \frac{\ln\big((1-f)+fx^2\big)}{1+x^2}\, dx= 2\ln\big(\sqrt{f}+\sqrt{1-f}\big),
$$
where the second equality follows from \cite[Formula (4.295.1)]{Gradshtein-Ryzhyk}, so that
$$
f^*=\frac{1}{2},\ g_r(f^*)=\ln 2\approx 0.69.
$$
Now  $\SR_r(f)=g_r(f)\big(\bar{\vr}_r(f)-g^2_r(f)\big)^{-1/2}$, where
$$
\bar{\vr}_r(f)=\frac{2}{\pi}\int_0^{+\infty} \frac{\ln^2\big((1-f)+fx^2\big)}{1+x^2}\, dx,
$$
and the integral can be evaluated numerically.  Note that, in this case, $\bE|r|=+\infty$,
$g_r(f)=g_r(1-f)>0$, $f\in (0,1)$, and, as $f\to 0+$, $g_r(f)\sim \bar{\vr}_r(f)\sim \sqrt{f}$. Indeed, the integral $\int_0^{+\infty} \ln^n(1+x^2) dx/x^2$ converges
for every  $n\geq 1$ so that, substituting $u=f^{1/2} x$,
$$
f^{-1/2}\int_0^{+\infty} \frac{\ln^n\big((1-f)+fx^2\big)}{1+x^2}\, dx \to  \int_0^{+\infty} \frac{\ln^n\big(1+u^2\big)}{u^2}\, du,\ f\to 0+.
$$
In particular, both $\vr_r$ and $\sqrt{\vr_r}$ have an inflection point.

Figure 1(b) presents the graphs of $f\mapsto g_r(f)$ and $f\mapsto \SR_r(f)$ when the distribution of returns is \eqref{Ch0}. We conclude that
\begin{align*}
f^{\circ}&=\underset{f}{{\mathrm{argmax}}}\;\SR_r(f)\approx 0.25, \ \SR_r(f^{\circ})\approx 0.41,\ g_r(f^{\circ})\approx 0.62;\\
f^{*}&=\underset{f}{{\mathrm{argmax}}}\;g_r(f) =0.5,\ \SR_r(f^*)\approx 0.38,\ g_r(f^*)=\ln 2\approx 0.69.
\end{align*}
Similar to the simple Bernoulli model, $f^*=0.5$, but
\begin{itemize}
\item The heavy right tail of $r$ results in a higher growth rate at and around $f^*$, but with smaller Sharpe ratio;
\item There is  a well-defined fractional Kelly criterion, corresponding to maximal asymptotic Sharpe ratio, as opposed to maximal asymptotic growth rate;
\item Maximizing asymptotic Sharpe ratio  reduces the  proportion $f$ from $0.5$ to about $0.25$, but leads to  moderate changes (about 10\%) in both $g_r$ and $\SR_r$.
\end{itemize}

As a consequence of heavy right tail of  \eqref{Ch0} and the resulting large values of $\vr_r$,
we find $f^{\Ri}=0$ for $\gamma\geq 1$. By comparison, if $\gamma=0.2$, then
$f^{\Ri} \approx 0.16$, which is close to $f^{\circ}$.

% This raises the question of how to select the parameter $\gamma$. Researchers often estimate $\gamma$ by analyzing data on individual decision making under uncertainty, using sources such as surveys or experimental studies. Maybe more

}
This concludes Example \ref{exm-CM}.
\end{exm}

\begin{exm}
\label{exm-T3}
    {\rm Here is another model where $f^*$ is computable in closed form.
    Consider
    \begin{equation*}
       r =  T_3^2 - 1,
    \end{equation*}
    where $T_3$ has  Student's $t$-distribution with $3$ degrees of freedom.
    Then
    \begin{equation*}
     \begin{split}
    g_r(f) &= \frac{4}{\pi \sqrt{3}} \int_0^{+\infty} \ln \left( (1 - f) + f x^2 \right) \left( 1 + \frac{x^2}{3} \right)^{-2} \, dx \\
    &= 2 \ln \left( \sqrt{1 - f} + \sqrt{3f} - \frac{\sqrt{f}}{\sqrt{\frac{1 - f}{3}} + \sqrt{f}} \right),
    \end{split}
    \end{equation*}
    where the second equality follows from \cite[Formula (4.295.25)]{Gradshtein-Ryzhyk}.
    By using the SymPy Python library, which performs symbolic mathematics, we find that
    $$
    f^* = \frac{7 + \sqrt{21} - \sqrt{6\sqrt{21} - 18}}{16} \approx 0.531,\ g_r(f^*)\approx 0.52.
    $$
    The function $f\mapsto \SR_r(f)$ can be evaluated numerically, and the corresponding graphs
    look similar to Figure 1(b). In particular, the maximal value of $\SR_r(f)$ is approximately $1.2$
    and is achieved at the point $f^{\circ}\approx 0.2$.
    }

    This concludes Example \ref{exm-T3}.
\end{exm}

Here are some general observations related to Examples \ref{exm-CM} and \ref{exm-T3}:
\begin{itemize}
 \item While the Cauchy distribution is  in the Student family
    $T_n$,  corresponding to $n=1$ degree of freedom, it appears that considering $r =  T_n^2 - 1$
    for other values of $n$ does not lead to closed-form expressions for $g_r(f)$ and $f^*$;
    \item If $r=T_n^2-1$, then, for $n\geq 3,$ we have $\bE [r] = 2/(n-2)$, meaning that the edge is decreasing with $n$ and becomes zero in the limit $n\to \infty$.
    \end{itemize}

% The examples above suggest the occurrence of a phase transition. We show that the mapping $ \gamma \mapsto f^{\Ri}(\gamma) $ is continuous.

% \begin{lem}\label{lem-lim-argmax}
%     Let $X$ be a compact set and let $u_n(x):X\to \bR$ be sequence of continuous functions that converge to $u(x)$ uniformly in $x$. Then,
%     \begin{equation*}
%         \lim_{n\to\infty}\arg \max_{x\in X} u_n(x) = \arg \max_{x\in X} u(x)
%     \end{equation*}
% \end{lem}

% The following proposition follows directly from the pervious Lemma.
% \begin{prop}
%     $\gamma \mapsto f^\Ri(\gamma)$ is an decreasing continuous function with
%     \begin{equation*}
%         \Im(f^\Ri)=(0,f^\Ri], \quad f^\Ri(0) = f^*, \quad \lim_{\gamma \to \infty}f^\Ri(\gamma) = 0.
%     \end{equation*}
% \end{prop}

% Define power utility function as
% \begin{equation}\label{eq:power-utility}
%     U_{\eta}(x) = \frac{x^\eta}{\eta}, \quad \eta<1,
% \end{equation}
% where the case $\eta = 0$ corresponds to the logarithmic utility function.
% \begin{equation*}
%     f^{\eta}=\frac{\mu}{\sigma^2(1-\eta)}.
% \end{equation*}

% \begin{cor}

% \end{cor}

 \subsection{Beyond the Standard Model}

 Our analysis essentially relies on the existence of the functions $g_r$ and $\vr_r$ from \eqref{F(f)}. Integrability assumptions are necessary for this existence.
 If $\bE\big|\ln(1+fr)\big|=+\infty$, then it is impossible to define $g_r$: by \cite[Theorem 2]{Robbins-InfIID},
 given any  real numbers $a_n,\ n\geq 1,$ and $b$,
 $$
 \bP\left(\lim_{n\to \infty} \frac{1}{a_n}\sum_{k=1}^n \ln (1+fr) = b\right) =0.
 $$
 Similarly, if $g_r(f)$ is defined, but $\bE[\ln(1+fr)]^2=+\infty$, then $\vr_r=+\infty$,
 $\SR_r=0$, the
 fluctuations of $\ln W^f_t$ around $tg_r(f)$ are of order larger than $\sqrt{t}$ and are no longer Gaussian; as a result, the corresponding  analysis falls outside the scope of our discussion.
 On the other hand, some models in continuous time can have $\vr_r=0$;
 see Examples \ref{exm-CIRv} and
 \ref{exm-logistic-price} below.

 We will now relax the independence assumption by allowing $\{r_k,\ k\geq 1\}$ to be a strictly stationary ergodic sequence, with random variable $r$
 representing the invariant distribution. In particular, each $r_k$ has the same distribution as $r$.
Similar to \eqref{rate},  the strong law of large numbers,
 $$
 \lim_{n\to \infty}  \frac{1}{n}\sum_{k=1}^n  \ln(1+fr_k)=\bE[\ln(1+fr)]=g_r(f),
 $$
   follows from a suitable  ergodic theorem \cite[Section 20.2]{Klenke}. The
   corresponding central limit theorem,
    \begin{equation}
 \label{CLT-MC-Br}
 \lim_{n\to \infty} \frac{1}{\sqrt{n}}\sum_{k=1}^n \big(\ln(1+fr_k)-g_r(f)\big)
 \wc \mathcal{N}\big(0,\widetilde{\vr}_r(f)\big),
 \end{equation}
 with
 \begin{equation}
 \label{var-MC-Br}
 \widetilde{\vr}_r(f)=\vr_r(f)+2\sum_{k=2}^{\infty} \mathrm{Cov}\big(\ln(1+fr_k), \ln(1+fr_1)\big),\ \vr_r(f)=\mathrm{Var}[\ln(1+fr)],
 \end{equation}
 requires additional conditions of {\em weak dependence}. Often, existence of higher-order moments of $\ln(1+r)$ is also
 required; cf. \cite[Theorem 7.3.1]{Kurtz-MP}. The following  two classes of models are {\em exponentially mixing} and, in particular, satisfy  \eqref{CLT-MC-Br}.
 \begin{itemize}
 \item The sequence $\{r_k,\ k\geq 1\}$ is  a finite-state Markov chain that is stationary, irreducible, and
 aperiodic \cite[Problem 7.10(a)]{Kurtz-MP}.  For a more detailed discussion, see, for example, \cite[Chapter 16]{ST-MC}.
 \item The sequence $\{r_k,\ k\geq 1\}$ is defined by
 \begin{equation}
 \label{AR-lin}
 r_{k+1}=ar_k+\xi_{k+1}, \ \ |a|<1,
 \end{equation}
 with iid $\xi_k$ having a  pdf  supported on $\bR$ and satisfying $\bE|\xi_k|<\infty$; the corresponding invariant distribution is
 $$
 r=\sum_{k=0}^{\infty}a^k\zeta_k,
 $$
 where the sequence $\{\zeta_k,\ k\geq 0\}$ is an independent copy of  $\{\xi_k,\ k\geq 1\}$.
 This is a particular case of the general model
 \begin{equation}
 \label{NL-AR1}
 r_{k+1}=a(r_k)+b(r_k)\xi_{k+1};
 \end{equation}
 see \cite[Theorem 2.1]{Erg-1995} for details.
 \end{itemize}

 With  \eqref{CLT-MC-Br} in place,  the asymptotic Sharpe ratio $\widetilde{\SR}_r$ becomes
 $$
 \widetilde{\SR}_r(f)=\frac{g_r(f)}{\sqrt{\widetilde{\vr}_r(f)}}.
 $$
 If sufficiently many of the  random variables $\ln(1+fr_k),\  \ln(1+fr_m),\ k\not=m,$  are negatively correlated, then $ \widetilde{\vr}_r(f)<\vr_r(f)$ is a possibility,
  resulting in a larger asymptotic Sharpe ratio compared to the iid case with the same marginal distribution $r$.

 For example, consider a generalization of the simple Bernoulli model, when the return $r_{k+1}$ on step $k+1$ depends on $r_k$ as follows:
 \begin{equation}
 \label{MC-Br}
 \begin{split}
 &\bP(r_{k+1}=1|r_k=1)=p,\ \bP(r_{k+1}=-1|r_k=1)=1-p, \\
 & \bP(r_{k+1}=1|r_k=-1)=1-q,\ \bP(r_{k+1}=-1|r_k=-1)=q,
 \end{split}
 \end{equation}
 for some $p,q\in (0,1)$. The iid case corresponds to $p=1-q$. In other words, the return sequence $r_k,\ k\geq 1$, is a two-state Markov chain
 with transition probability matrix
 $$
P= \left(\begin{array}{cc}
 p & 1-p\\
 1-q & q
 \end{array}\right).
 $$
 The unique invariant distribution is
 $$
 \bP(r=1)=\frac{1-q}{2-p-q},\ \bP(r=-1)=\frac{1-p}{2-p-q},
 $$
 so that
 $$
 f^*=\frac{p-q}{2-p-q}=\bE[r],
 $$
 and, because $p-q<2-p-q$,  the NS-NL condition is the same as the edge condition: $p>q$.
 If $r_1$ has the same distribution as $r$, then the sequence $\{r_k,\ k\geq 1\}$ is strictly stationary, ergodic, and exponentially mixing:
  the $n$-step transition matrix $P^n$ satisfies
 $$
 P^n=\frac{1}{2-p-q}\left(\begin{array}{cc}
  {1-q}  & {1-p} \\
  {1-q}  & {1-p}
 \end{array}\right)+ \frac{\rho^n}{2-p-q}
 \left(\begin{array}{cc}
 {1-p}  & - {1+p} \\
 - {1+q}  & {1-q}
 \end{array}\right)
 $$
with $\rho=p+q-1\in (-1,1).$ In particular, \eqref{CLT-MC-Br} holds, and,  by \eqref{var-bern},
 $$
 \vr_r(f)=\frac{(1-p)(1-q)}{(2-p-q)^2}\,\ln^2\left(\frac{1+f}{1-f}\right).
 $$
 Using the expression for $P^n$, we compute
   \begin{equation}
 \label{var-MC-Br1}
 \widetilde{\vr}_r(f)=\left(1+\frac{2\rho}{1-\rho}\right)\vr_r(f)=\left(1+\frac{2(p+q-1)}{2-p-q}\right)\vr_r(f)=\frac{p+q}{2-p-q}\,\vr_r(f).
 \end{equation}
As a result,  $ \widetilde{\vr}_r<\vr_r$ if, and only if, $p+q<1$, so that, compared with the iid case, one can have the same asymptotic rate of return, for the
same optimal strategy $f^*$,  but with a bigger asymptotic Sharpe ratio
$$
\widetilde{\SR}_r(f^*)=\frac{g_r(f^*)}{\sqrt{\widetilde{\vr}_r(f^*)}}.
$$

As a concrete numerical example, consider $p=17/24$, $q=3/24=1/8$. Then the stationary distribution corresponds to the simple Bernoulli model with $\bP(r=1)=3/4$,
so that $f^*=0.5$ and $g_r(f^*)\approx 0.131$. On the other hand, because $\rho=p+q-1=-1/6<0$, we have
$\widetilde{\vr}_r(f^*)=5\vr_r(f^*)/7$ and $\widetilde{\SR}_r(f^*)=\SR_r(f^*)\sqrt{7/5}\approx 1.183\cdot\SR_r(f^*)\approx 0.320$.
From Figure 1, we see that  the same increase of $\SR_r$ in the iid case is
achieved by decreasing $f$ from $0.50$ to $0.43$, leading to a 25\% decrease of the asymptotic growth rate.

Extending \eqref{TtoG-A} and \eqref{TtoG-CLT} to dependent case requires additional conditions that are not easily verifiable in our setting; cf. \cite{Gut-dt-d}.

Yet another extension of the standard model is multiple bets/multiple investment options: given the iid random
 vectors $\bld{r}_k=(r^{(1)}_k,\ldots, r^{(m)}_k)^{\top}\in \bR^m$ of returns, $r^{(\ell)}_k\geq -1$,  and the corresponding
strategy vector $\bld{f}=(f^{(1)},\ldots, f^{(m)})^{\top}$,
 \begin{equation}
 \label{wealth2-md}
  W_n^{\bld{f}} =  \prod_{k=1}^n\big(1 + \bld{f}\cdot \bld{r}_k\big), \ n=1,2,\ldots.
 \end{equation}
 The NS-NL condition becomes
 \begin{equation}
 \label{md-NLNS}
 f^{(\ell)}\geq 0,\ \sum_{\ell=1}^m f^{(\ell)}\leq 1.
 \end{equation}
 With minor modifications, some of the results for \eqref{wealth2-md} are identical to those for \eqref{wealth2}: see, for example, \cite[Lemmas 1--3]{Kelly-Finkel}.
 In particular, the optimal allocation vector
 $$
\bld{f}^*= \arg\max_{\bld{f}}\lim_{N\to \infty} \frac{\ln W_N^{\bld{f}}}{N}= \arg\max_{\bld{f}} \bE\big[ \ln(1+\bld{f}\cdot\bld{r})\big],
$$
exists; it is unique if the random variables $(r^{(1)}_k,\ldots, r^{(m)}_k)$ are linearly independent.

 The standard model is a particular case of \eqref{wealth2-md}, with $m=2$, $\bld{r}_k=(r_k,0)^{\top}$, $\bld{f}=(f,1-f)^{\top}$.
More generally, if $\bld{r}_k=(r_k,r_0)^{\top}$, with simple Bernoulli $r_k$ and non-random $r_0\in (0,1)$ representing risk-free return, then
 $$
 f^*=p-q\left(1+\frac{2r_0}{1-r_0}\right)<p-q,
 $$
 so that, under the NS-NL condition, if $r_0$ is close enough to $1$, then $f^*=0$: with high risk-free return,   there is no reason to gamble.
 Most of the time, though, the optimal strategy $\bld{f}^*$ does not have a closed-form expression. Beside direct numerical approximation, some information about
  $f^*$ can be obtained from a high-frequency version of the model, which is the subject of the next section.

\section{High-Frequency Compounding}
\label{sec:HFC}

In deterministic setting,   high-frequency compounding leads to exponential growth
$$
\lim_{n\to \infty} \left(1+\frac{r}{n}\right)^{\lfloor nt \rfloor }=e^{rt}.
$$
A slightly more careful analysis shows that the sequence
$$
W^{(n)}_{N}=\left(1+\frac{r}{n}\right)^{nN},\  n,N\geq 1,
$$
represents high-frequency compounding over $N$ unit steps, and
\begin{equation}
\label{Det-CT}
W^{(n)}_{t}=\left(1+\frac{r}{n}\right)^{\lfloor nt \rfloor },\ n\geq 1,\ t>0,
\end{equation}
is a continuous time interpolation of $W^{(n)}_{N}$: $W^{(n)}_{t}=W^{(n)}_{N}$ if $\lfloor nt \rfloor=nN$.

In  stochastic setting, high-frequency compounding can simplify  computation of the optimal allocation.
As an illustration, consider  high-frequency version of  \eqref{wealth2-md} under an additional assumption  that $\bE[|\bld{r}|^2]<\infty$.
If $\bld{\mu}=\bE \bld{r}$ and $Q$ is the covariance matrix of $\bld{r}$,  then
$$
\bld{r}\wc \bld{\mu} + \bld{\xi}
$$
for a suitable random vector $\bld{\xi}$ with mean zero and covariance matrix $Q$.
Now, let
$$
\bld{r}_n = \frac{\bld{\mu}}{n}+\frac{\bld{\xi}}{\sqrt{n}},\ n\geq 1,
$$
let $\bld{r}_{n,i},\ i\geq 1,$ be iid copies of $\bld{r}_n$, and define
\begin{equation}
\label{HF-vec}
W^{n,\bld{f}}_{N} = \prod_{k=1}^N \prod_{\ell=1}^n \Big(1+\bld{f}\cdot\bld{r}_{n,(k-1)n+\ell}\Big).
\end{equation}

\begin{prop}
\label{prop-MdAppr}
Assume that the matrix $Q$ is non-singular,  $\bE|\bld{r}|^3<\infty$, and the vector $Q^{-1}\bld{\mu}$ satisfies the NS-NL condition \eqref{md-NLNS}.
 Then, for every $n\geq 1$,   there exists a unique
\begin{equation}
\label{Md-opt}
\bld{f}^*_n=\arg\max_{\bld{f}}\lim_{N\to \infty} \frac{\ln W^{n,\bld{f}}_{N}}{N}
\end{equation}
and
\begin{equation}
\label{MdAppr}
\bld{f}^*_n= Q^{-1}\bld{\mu}+O(n^{-1/2}),\ \ n\to \infty.
\end{equation}
\end{prop}

\begin{proof}
We have
$$
\ln W^{n,\bld{f}}_{N}=\sum_{k=1}^N \left(\sum_{\ell=1}^n \ln \left(1+\frac{\bld{f}\cdot\bld{\mu}}{n}+\frac{\bld{f}\cdot\bld{\xi}_{(k-1)n+\ell}}{\sqrt{n}}\right)\right)
$$
and
\begin{equation}
\label{HF-a1}
\lim_{N\to \infty} \frac{\ln W^{n,\bld{f}}_{N}}{N}
= \bE \left(\sum_{\ell=1}^n \ln \left(1+\frac{\bld{f}\cdot\bld{\mu}}{n}
+\frac{\bld{f}\cdot\bld{\xi}_{\ell}}{\sqrt{n}}\right)\right)
=n\bE \ln \left(1+\frac{\bld{f}\cdot\bld{\mu}}{n}+\frac{\bld{f}\cdot\bld{\xi}}{\sqrt{n}}\right).
\end{equation}
If $\bE|\bld{r}|^3<\infty$, then the Taylor formula implies
\begin{equation}
\label{HF-a2}
n\bE \ln \left(1+\frac{\bld{f}\cdot\bld{\mu}}{n}+\frac{\bld{f}\cdot\bld{\xi}}{\sqrt{n}}\right)=
\bld{f}\cdot\bld{\mu}-\frac{1}{2}\,\bld{f}\cdot(Q\bld{f})+O(n^{-1/2}),\ \ n\to \infty.
\end{equation}
If   the matrix $Q$ is non-singular then the random variables $r^{(1)},\ldots, r^{(m)}$ are
linearly independent so that \eqref{Md-opt} follows by \cite[Lemmas 1]{Kelly-Finkel}.
On the other hand,
$$
 \arg\max_{\bld{f}}\Big(\bld{f}\cdot\bld{\mu}-\frac{1}{2}\,\bld{f}\cdot(Q\bld{f})\Big)=  Q^{-1}\bld{\mu}.
 $$
Then \eqref{MdAppr}  follows after comparing \eqref{HF-a1} and \eqref{HF-a2}.
\end{proof}

\begin{rem}
\label{rem-HF-md}
Similar to \eqref{HF-a2},
$$
n\mathrm{Var}\left[\ln \left(1+\frac{\bld{f}\cdot\bld{\mu}}{n}+\frac{\bld{f}\cdot\bld{\xi}}{\sqrt{n}}\right)\right] = \bld{f}\cdot(Q\bld{f})+O(n^{-1/2}),\ \ n\to \infty.
$$
Then, for $W^{n,\bld{f}}_{N}$ defined in \eqref{HF-vec}, we have the following versions of the LLN and the CLT:
\begin{align*}
&\lim_{N\to \infty} \frac{\ln W^{n,\bld{f}}_{N}}{N} = g^n_{\bld{r}}(\bld{f}) \ \ {\text{with probability one}},\\
&\lim_{N\to \infty} \sqrt{N }\left(\frac{\ln W^{n,\bld{f}}_{N}}{N} - g^n_{\bld{r}}(\bld{f})\right) = \mathcal{N}\big(0,\vr^n_{\bld{r}}(\bld{f})\big) \ \
 {\text{in distribution}},\\
 \end{align*}
 and
 \begin{align*}
&g^n_{\bld{r}}(\bld{f}) =\bld{f}\cdot\bld{\mu}-\frac{1}{2}\,\bld{f}\cdot(Q\bld{f})+O(n^{-1/2}), \ n\to \infty,\\
&v^n_{\bld{r}}(\bld{f})=\bld{f}\cdot(Q\bld{f})+O(n^{-1/2}),\ \ n\to \infty.
\end{align*}
\end{rem}

 Proposition \ref{prop-MdAppr} and Remark \ref{rem-HF-md} investigate $W^{n,\bld{f}}_{N}$ in the limit $\lim_{n\to \infty} \lim_{N\to \infty}$.
Accordingly, our next step is  to understand the behavior of $\lim_{n\to \infty} W^{n,\bld{f}}_{N}$ for fixed $N$;  for simplicity, we restrict the discussion to the scalar  $f$.

Consider the following stochastic version of \eqref{Det-CT}:
\begin{equation}
\label{HF-W}
W_t^{n,f} = \prod_{k=1}^{\lfloor nt \rfloor}\big(1+fr_{n,k}\big),
\end{equation}
where
\begin{equation}
\label{HF-1}
r_{n,k}=\frac{\mu}{n}+\frac{\sigma }{\sqrt{n}}\xi_{k}
\end{equation}
and $\xi_{k},\ k\geq 1,$ are iid random variables with zero mean, unit variance, and finite third moment.
Let $B=B_t,\ t\geq 0,$ be a standard Brownian motion on
a stochastic basis $(\Omega, \mathcal{F}, \{\mathcal{F}_t\}_{t\geq 0},\mathbb{P})$
satisfying the usual conditions \cite[Section I.1]{Protter}.
Using the Taylor formula,
$$
\ln W_t^{n,f} = \sum_{k=1}^{\lfloor nt \rfloor}\ln  \big(1+fr_{n,k}\big) = \left(f\mu-\frac{f^2\sigma^2}{2}\right) t +
 \frac{\sigma f}{\sqrt{n}}\sum_{k=1}^{\lfloor nt \rfloor}\xi_{n,k}+O(n^{-1/2}),\ \ n\to \infty;
 $$
 in the limit $n\to \infty$,  the Donsker Invariance Principle \cite[Theorem 21.43]{Klenke} leads to geometric Brownian motion with drift:
\begin{equation}
\label{eq:GBM-D}
\lim_{n\to \infty} W_t^{n,f} \wcL \exp\left(\left(f\mu-\frac{f^2\sigma^2}{2}\right)t+f\sigma B_t\right).
\end{equation}
The convergence in \eqref{eq:GBM-D} is weak in the space $\mathcal{C}(\bR_+)$ of continuous functions on $\bR_+=[0,+\infty)$ and uniform in $f$ over compact subsets of $(0,1)$.
For details, see \cite[Theorem 4.1]{L-P}.

A minor technical complications, that the function $t\mapsto W_t^{n,f}$ might not be continuous, is fixed with an asymptotically negligible correction: the function
$$
t\mapsto \exp\left(\ln W_t^{n,f}+\big(n t-\lfloor n t \rfloor\big)\ln \big(1+ fr_{n,\lfloor n t \rfloor+1}\big)\right),
$$
based on the linear interpolation of $\ln W_t^{n,f}$, is continuous. To avoid unnecessary complications in the formulas, we will not write this correction
while discussing weak convergence in $\mathcal{C}(\bR_+)$.

Because condition $r_{n,k}\geq -1$ puts additional restrictions on the random
variables $\xi_{k}$ in representation \eqref{HF-1}, beyond the
usual conditions of the Donsker invariance principle, it is often more convenient to consider {\tt geometric high-frequency} compounding: the process \eqref{HF-W} with
\begin{equation}
\label{expo-r}
r_{n,k}=\exp\left(\frac{\mu-(\sigma^2/2)}{n}+\frac{\sigma}{\sqrt{n}}\,\xi_{k}\right)-1
\end{equation}
instead of \eqref{HF-1}. The term $-\sigma^2/(2n)$ in \eqref{expo-r}  ensures that, as $n\to \infty$,
$$
 \bE[r_{n,k}]=\frac{\mu}{n} + O(n^{-3/2}),\ \mathrm{Var}[r_{n,k}]=\frac{\sigma^2}{n}+O(n^{-3/2}),
 $$
 that is, \eqref{expo-r} is asymptotically equivalent to \eqref{HF-1}; we refer to this term as
  the {\tt It\^{o} correction}.
If $\xi_k$ are iid, then \eqref{eq:GBM-D} holds \cite[Section 4.1]{L-P}.

We will now drop the independence assumption and consider stationary ergodic sequences. The following definition will simplify the
presentation.

\begin{dfn}
	\label{def-erg}
A sequence  $\xi=\{\xi_{k}, k\in \mathbb{Z}\}$ is called $\varrho$-normal if
\begin{align}
\notag
&\text{the sequence is strictly stationary and ergodic;}\\
\label{mean-var}
&\bE[\xi_{k}]=0,\ \bE|\xi_{k}|^2=1;\\
 \label{cov-conv}
 & 1+2\sum_{k=2}^{\infty} \bE[\xi_{k}\xi_{1}]=\varrho^2\in (0,+\infty);\\
\label{Donsker-D}
&\lim_{n\to \infty} \frac{1}{\sqrt{n}}\sum_{k=1}^{\lfloor nt \rfloor} \xi_{k} \wcL \varrho  B_t\ \ {\rm weakly\ \ in}\ \ \mathcal{C}(\bR_+).
\end{align}
\end{dfn}

By \cite[Theorem VIII.3.97]{LimitTheoremsforStochasticProcesses}, a  strictly stationary and ergodic sequence is $\varrho$-normal if
$$
\sum_{k=1}^{\infty} \left(\bE\big|\bE(\xi_k|\xi_m,\ m\leq 0)\big|^2\right)^{1/2}<\infty.
$$
To state a more sophisticated sufficient condition, define the sigma-algebras
$$
\mathcal{F}^{\xi}_0=\sigma(\xi_k,\ k\geq 0), \ \  \mathcal{F}^{\xi,n}=\sigma(\xi_k,\ k\geq n),
$$
 and the {\em mixing coefficient}
 $$
 \boldsymbol{\varphi}_q(m)=\sup_{A\in \mathcal{F}^{\xi,m}} \Big(\bE \big|\mathbb{P}(A|\mathcal{F}^{\xi}_0)-\mathbb{P}(A)\big|^q\Big)^{1/q},\ q\geq 1.
 $$
 The following result holds.
\begin{prop}
\label{Kurtz}
Assume that $\xi=\{\xi_{k}, k\in \mathbb{Z}\}$ is stationary, with $\bE[\xi_{k}]=0,\ \bE|\xi_{k}|^2=1$ and
\begin{align}
\label{m-del}
&\bE|\xi_{1}|^p<\infty, \ p>2;\\
\label{mixing-p}
&\sum_{m=1}^{\infty} \big(\boldsymbol{\varphi}_q(m)\big)^{(p-2)/(p-1)}<\infty,\ \frac{1}{p}+\frac{1}{q}=1.
\end{align}
Then the sequence $\xi$ is $\varrho$-normal.
\end{prop}

\begin{proof} See \cite[Theorem 7.3.1]{Kurtz-MP}. In particular, \eqref{mixing-p} implies ergodicity of $\xi$.
\end{proof}

\begin{thm}
\label{prop0}
Let $\xi$ be $\varrho$-normal sequence and
\begin{equation}
\label{moment3}
\bE|\xi_{1}|^3<\infty.
\end{equation}
Define $r_{n,k}$ and $W_t^{n,f}$ by \eqref{expo-r} and \eqref{HF-W}, respectively.

Then, for every $f\in (0,1)$,   the sequence of processes
$ \big(W_t^{n,f},\ n\geq 1,\ t\geq 0\big)$ converges weakly in $\mathcal{C}(\bR_+)$ to the process
$$
W^f_t=\exp\left(\left(f\mu-\frac{f^2\sigma^2}{2}\right)t+f\sigma \varrho B_t\right) ,\ t \geq 0,
$$
and the convergence is uniform in $f$ on compact subsets of $(0,1)$.
\end{thm}

\begin{proof} For $f\in (0,1)$, define the function
$$
F(x)=\ln(1-f+fe^x),\ x\in \bR,
$$
so that $F(0)=0,\ F'(0)=f$, $F''(0)=f-f^2$, and, given a compact set $A\subset (0,1)$,
\begin{equation}
\label{F-bd3}
\sup_{x\in \bR, f\in A}|F'''(x)|=C_A
\end{equation}
for some number $C_A$ depending only on the set $A$. Also, let
$$
X_{n,k}=\frac{\mu-(\sigma^2/2)}{n}+\frac{\sigma}{\sqrt{n}}\,\xi_{k},  \
U^n_t= \sum_{k=1}^{\lfloor nt \rfloor} \left(F'(0)X_{n,k}+\frac{1}{2}F''(0)X^2_{n,k}\right),\
V^n_t=\sum_{k=1}^{\lfloor nt \rfloor} F(X_{n,k}).
$$
The statement of the theorem is equivalent to
$$
\lim_{n\to \infty}  V^n_t \wcL \left(f\mu-\frac{f^2\sigma^2}{2}\right)t+f\sigma \varrho B_t.
$$
By assumption,
$$
\lim_{n\to \infty} \frac{1}{n}\sum_{k=1}^{\lfloor nt \rfloor} \xi_{k}=0,\  \lim_{n\to \infty} \frac{1}{n}\sum_{k=1}^{\lfloor nt \rfloor} |\xi_{k}|^2=t,
$$
both with probability one, and therefore
\begin{align*}
\lim_{n\to \infty}\sum_{k=1}^{\lfloor nt \rfloor}X_{n,k}&=\lim_{n\to \infty}\sum_{k=1}^{\lfloor nt \rfloor}\frac{\mu-(\sigma^2/2)}{n}+\lim_{n\to \infty}
\frac{\sigma}{\sqrt{n}}\sum_{k=1}^{\lfloor nt \rfloor}\xi_{k}\wcL\left(\mu-\frac{\sigma^2}{2}\right) t + \sigma\varrho B_t,\\
\lim_{n\to \infty}\sum_{k=1}^{\lfloor nt \rfloor}X^2_{n,k}&=\lim_{n\to \infty}\frac{\sigma^2}{n}\sum_{k=1}^{\lfloor nt \rfloor}\xi_k^2=\sigma^2 t,
\end{align*}
so that
$$
\lim_{n\to \infty}  U^n_t \wcL \left(f\mu-\frac{f^2\sigma^2}{2}\right)t+f\sigma \varrho B_t.
$$
It remains to show that, for every $T>0$ and a compact set $A\subset (0,1)$,
$$
\lim_{n\to \infty} \bE \sup_{0<t<T, f\in A} |U_t^n-V_t^n|=0.
$$
Using \eqref{F-bd3}, the Taylor formula, and $(x+y)^3\leq 4(x^3+y^3), \ x,y>0$ (a version of H\"{o}lder's inequality),
$$
|F(X_{n,k})-F'(0)X_{n,k}-\frac{1}{2}F''(0)X^2_{n,k}|\leq C_A|X_{n,k}|^3\leq 4C_A(\mu^3+\sigma^3)\left(\frac{1}{n^3}+\frac{|\xi_k|^3}{n^{3/2}}\right)
$$
so that
$$
\bE \sup_{0<t<T, f\in A}|U_t^n-V_t^n|\leq  \frac{4TC_A(\mu^3+\sigma^3)}{\sqrt{n}}\left(1+2\bE |\xi_1|^3\right),
$$
concluding the proof of Theorem \ref{prop0}.
\end{proof}

\begin{cor}
In the setting of Theorem \ref{prop0}, let
$$
g_n(f)=\lim_{t\to \infty} \frac{\ln W^{n,f}_t}{t}
$$
and
 \begin{equation}
 \label{fn*}
  f_n^*=\arg\max\limits_{f\in [0,1]}g_n(f).
  \end{equation}
If $0<\mu<\sigma^2$, then
\begin{equation}
\label{fn-conv}
\lim_{n\to \infty} f_n^*=\frac{\mu}{\sigma^2},\ \qquad \ \lim_{n\to \infty} g_n(f_n^*)=\frac{\mu^2}{2\sigma^2}.
\end{equation}
  \end{cor}

\begin{rem}
With some extra effort, condition \eqref{moment3}  can be weakened to
\eqref{m-del}.
\end{rem}

\begin{rem}
A strictly stationary sequence $\xi=\xi_k,\ k=1,2,3,\ldots$ can be embedded into a stationary sequence
$\widehat{\xi}=\widehat{\xi}_k,\ k\in \mathbb{Z}$; cf. \cite[Remark 7.3.2(a)]{Kurtz-MP}.
\end{rem}

\section{Continuous Time  Compounding}
\label{sec:CT}

Continuous time   compounding with rate, or return, process $R$ and constant proportion $f$ of the invested amount can be modeled by
\begin{equation}
\label{wealth1-ct}
dW_{t}^f=fW_t^fdR_{t},\ t>0,
\end{equation}
for a suitable semi-martingale $R=R_t$ on a stochastic basis
\begin{equation*}
%\label{st-b}
\mathbb{F}=\Big(\Omega, \mathcal{F},\ \{\mathcal{F}_t\}_{t\geq 0},
\mathbb{P}\Big)
\end{equation*}
satisfying the usual conditions \cite[Section I.1]{Protter}.
With the assumption   $W_0^f=1$,  \eqref{wealth1-ct} is equivalent to the  integral equation
\begin{equation}
\label{wealth1-ct-i}
W_{t}^f=1+f\int_0^tW_s^fdR_{s}
\end{equation}

If the trajectories of $R$ are continuous with quadratic variation $\langle R \rangle$, and \eqref{wealth1-ct-i} is in the It\^{o} sense, then
\begin{equation}
\label{DDE-2-cont}
W^f_t=\exp\left(fR_t-\frac{f^2}{2}\langle R\rangle_t\right);
\end{equation}
see \cite[Theorem II.36]{Protter}.
In particular,
\begin{equation}
\label{DDE-3-cont}
\frac{\ln W^f_t}{t} = \frac{fR_t}{t}-\frac{f^2\langle R\rangle_t}{2t}.
\end{equation}

To proceed, recall \cite[Definition 2.1.1]{LSh-M} that a continuous semi-martingale is a sum of a continuous process $A$ of bounded variation and a continuous local martingale
$M$:
$$
R_t=A_t+M_t,
$$
with $ \langle R\rangle_t=\langle M\rangle_t$. If all the trajectories   $t\mapsto A_t$ and $t\mapsto \langle M\rangle$ are absolutely continuous with respect to the
Lebesgue measure, then
\begin{equation}
\label{R-simple-s}
R_t =  \int_0^t r_s\, ds + \int_0^t \sqrt{v_s}\, d B_s
\end{equation}
for some $\mathcal{F}_t$-adapted, locally (in time) integrable processes $r=r_t$ and $v=v_t\geq 0$; cf. \cite[Theorem 5.12]{LSh-S1}.
Accordingly, in what follows, we assume that  \eqref{R-simple-s} holds.

The next step is to derive conditions on the processes $r$ and $v$ that lead to continuous time analogs of the LLN  \eqref{rate} and CLT \eqref{eq:LLN-g2}, namely,
\begin{align}
\label{CT-gr}
&\lim_{t\to \infty}\frac{\ln W^f_t}{t} = g_R(f) \ \ {\text{with probability one,}}\\
\label{CLT-W-f}
&\lim_{t\to \infty} \frac{1}{\sqrt{t}}\left(\ln W^f_t-t \,g_R(f)\right) = \mathcal{N}\big(0,\vr_R(f)\big) \ \ {\text{in distribution}},
\end{align}
and to identify the corresponding functions $g_R$ and $\vr_R$.

We will use the following conditions on the processes $r$ and $v$ in \eqref{R-simple-s}.

{\sc Condition (C1).} There exist positive numbers $\mu$ and $\sigma$ such that,
with probability one,
\begin{equation}
\label{ct-erg1}
\lim_{t\to \infty} \frac{1}{t}\int_0^tr_s\, ds =\mu,\ \quad \ \lim_{t\to \infty} \frac{1}{t}\int_0^t v_s\, ds =\sigma^2.
\end{equation}

Next, with $\mu$ and $\sigma$ from \eqref{ct-erg1}, define
\begin{equation}
\label{aux-hat}
\widehat{r}(t)=\frac{1}{\sqrt{t}}\int_0^t (r_s-\mu)ds,\ \ \widehat{\sigma}(t)=\frac{1}{\sqrt{t}}\int_0^t \sqrt{v_s}\, dB_s,\ \
\widehat{v}(t)=\frac{1}{\sqrt{t}}\int_0^t (v_s-\sigma^2)ds.
\end{equation}

{\sc Condition (C2).} As $t\to \infty$, the three-dimensional process
$$
t \mapsto \big(\widehat{r}(t),\widehat{\sigma}(t),\widehat{v}(t)\big)
$$
converges in distribution to a Gaussian random vector with
mean zero and covariance matrix ${Q}$.

% To this end, we define the following auxiliary processes
% \begin{equation}
% \label{aux-hat}
% \widehat{r}(t)=\frac{1}{\sqrt{t}}\int_0^t (r_s-\mu)ds,\ \ \widehat{\sigma}(t)=\frac{1}{\sqrt{t}}\int_0^t \sqrt{v_s}\, dB_s,\ \
% \widehat{v}(t)=\frac{1}{\sqrt{t}}\int_0^t (v_s-\sigma^2)ds.
% \end{equation}

\begin{thm}
\label{thm-R-simple-s}
Under {\sc (C1)}, convergence \eqref{CT-gr} holds with
\begin{equation}
\label{gR00}
g_R(f)= f\mu-\frac{f^2\sigma^2}{2}.
\end{equation}
Under  {\sc (C1), (C2)}, convergence
\eqref{CLT-W-f} holds with
$\vr_R(f)=\widehat{\boldsymbol{f}}{Q}\widehat{\boldsymbol{f}}^{\top}$ and row vector $\widehat{\boldsymbol{f}}=(f,f,-f^2/2)$.
\end{thm}

\begin{proof} If $\langle R \rangle_t = \int_0^t v_s\,ds$, then equality  \eqref{DDE-3-cont} becomes
\begin{equation}
\label{LogW-erg}
\frac{\ln W^f_t}{t} = \frac{f}{t}\int_0^tr_s\, ds - \frac{f^2}{2t}\int_0^t v_s\, ds+\frac{f}{t} \int_0^t \sqrt{v_s}\, dB_s.
\end{equation}
By the strong law of large numbers for square integrable martingales \cite[Corollary 1 to Theorem 2.6.10]{LSh-M},
$$
\lim_{t\to \infty} \frac{1}{t} \int_0^t \sqrt{v_s}\, dB_s=0 \ \ {\text{with probability one.}}
$$

Next,  \eqref{LogW-erg} implies
\begin{equation}
\label{R-simple-s-vt}
\frac{1}{\sqrt{t}}\left(\ln W^f_t-t \,g_R(f)\right)=\frac{f}{\sqrt{t}}\int_0^t (r_s-\mu)ds+\frac{f}{\sqrt{t}}\int_0^t \sqrt{v_s}\, dB_s -
\frac{f^2}{2\sqrt{t}}\int_0^t (v_s-\sigma^2)ds,
\end{equation}
that is,
$$
\frac{1}{\sqrt{t}}\left(\ln W^f_t-t \,g_R(f)\right)=f\widehat{r}(t)+f\widehat{\sigma}(t)-\frac{f^2}{2}\widehat{v}(t),
$$
and the joint convergence  leads to \eqref{CLT-W-f}.
\end{proof}

\begin{rem}
If  \eqref{gR00} holds, then
\begin{equation}
\label{CTW-f*}
 f^*\equiv \arg\max_f g_R(f)=\frac{\mu}{\sigma^2},\ \  g_R(f^*)=\frac{\mu^2}{2\sigma^2},
\end{equation}
and the NS-NL condition becomes
$$
0\leq \mu\leq \sigma^2.
$$
\end{rem}

The following special case of \cite[Theorem 7.4.1]{Kurtz-MP} can be used to verify
{\sc Condition (C2)}.

\begin{prop}
	\label{pro-joint-conv}
	Let $a_{k,t},\  k=1,\ldots, N,\ t\geq 0,$ be adapted square-integrable processes with the following property: for every $k,m$, there exists a real number $c_{km}$ such that
	$$
	\lim_{t\to \infty} \frac{1}{t}\int_0^t a_{k,s}a_{m,s}\, ds = c_{km}\ \ \ \ {\text{in probability}}.
	$$
	Let $B_{k,t},\ k=1,\ldots, N,\ t\geq 0,$ be an $N$-dimensional Gaussian process with mean zero and covariance $\bE [B_{k,t}B_{m,s}]=\rho_{km}\min(t,s)$, where $\rho_{km}\in [-1,1]$ and
		$\rho_{kk}=1$. Then, as $t\to \infty$,  the $N$-dimensional process
		$$
		\left(\frac{1}{\sqrt{t}}\int_0^t a_{1,s}dB_{1,s},\ldots,  \frac{1}{\sqrt{t}}\int_0^t a_{N,s}dB_{N,s}\right)
		$$
		converges in distribution to a Gaussian random vector with mean zero and covariance matrix $(c_{km}\rho_{km},\ k,m=1,\ldots, N)$.
\end{prop}

\begin{proof}
Let $b_n,\ n\geq 1$ be a sequence of positive real numbers converging to $+\infty$ and consider the sequence of $N$-dimensional martingales
$\boldsymbol{M}^{(n)}_t=\left(M^{(n)}_{k,t},\, k=1,\ldots, N\right),\ \ n\geq 1, \ \ t\geq 0,$ defined by
$$
M^{(n)}_{k,t}=\frac{1}{\sqrt{b_n}}\int_0^{b_nt}a_{k,s}\,dB_{k,s}.
$$
By assumption, for every $k$, $m$, and $t>0$,
$$
\langle M^{(n)}_k,M^{(n)}_m \rangle_t = \frac{1}{b_n}\int_0^{b_nt} a_{k,s}a_{m,s}\rho_{km}\, ds \to c_{km}\rho_{km}t,\ n\to \infty,
 $$
 in probability. Therefore, using \cite[Theorem 7.4.1]{Kurtz-MP}, we conclude that
 $$
 \lim_{n\to \infty} \boldsymbol{M}^{(n)} \wcL \boldsymbol{M},
 $$
 weakly in the space of continuous functions, and $\boldsymbol{M}_t=(M_{1,t},\ldots, M_{N,t}),\ t\geq 0,$ is a continuous square-integrable martingale with
 $$
 \langle M_k,M_m\rangle_t=c_{km}\rho_{km}t.
 $$
Using the multi-dimensional version of the L\'{e}vy characterization of the Brownian motion (which, in turn, can be considered a very special case of uniqueness for
martingale problems, such as \cite[Theorem 5.3.1]{Str-Var}), we conclude that  $\boldsymbol{M}_t$ is the  Gaussian vector with the desired distribution.
\end{proof}

\begin{exm}
\label{exm-GBM}
{\rm {
The easiest example of the process $R$ satisfying conditions of Theorem \ref{thm-R-simple-s} is
\begin{equation}
\label{R-BM}
R_t= \mu t + \sigma B_t,
\end{equation}
having constant $r_t=\mu$ and $v_t=\sigma^2$. In this case
\begin{equation}
\label{vR00}
\vr_R(f)=f^2\sigma^2.
\end{equation}
Combining \eqref{gR00} and \eqref{vR00}, we conclude that the corresponding asymptotic Sharpe ratio
$$
\SR_R(f)=\frac{g_R(f)}{\sqrt{\vr_R(f)}}=\frac{\mu}{\sigma}-\frac{f\sigma}{2},
$$
is a linear decreasing function of $f$, and
$$
\SR_R(f^*)=\frac{\mu}{2\sigma}.
$$
The corresponding optimal ridge strategy, maximizing
$$
\Ri_R(f,\gamma)= \mu f -\frac{f^2\sigma^2}{2}-\gamma f^2\sigma^2
$$
with respect to $f$, is
\begin{equation}
\label{fRi}
    f^{\Ri} = \frac{\mu}{\sigma^2(1+2\gamma)} = \frac{f^*}{1+2\gamma}.
\end{equation}
% and
% \begin{equation*}
%     \Ri_R(f^*) = \frac{\mu^2}{\sigma^2}(1/2 - \gamma) \quad  \text{and} \quad \Ri_R(f^\Ri) = \frac{\mu^2}{2\sigma^2(1 + 2\gamma)}.
% \end{equation*}
% See that when $\gamma>1/2$ then ridge coefficient at $f^*$ can be negative.
}}
This concludes Example \ref{exm-GBM}.
\end{exm}

The following is a direct extension of \eqref{fRi}.

\begin{prop}
    Let the return process $t\mapsto R_t$ be such that the asymptotic long-term growth rate and variance are given by
\begin{equation*}
g_R(f)= f\mu-\frac{f^2\sigma^2}{2}, \quad \vr_R (f) = f^2c^2,
\end{equation*}
for $c,\sigma,\mu>0$. Then the optimal ridge strategy is
\begin{equation*}
    f^{\Ri}(\gamma) = \arg\max\Ri_R(f,\gamma)= \frac{\mu}{\sigma^2(1+2\gamma(c/\sigma)^2)}.
\end{equation*}
\end{prop}

Next, let us discuss a connection between fractional Kelly strategies and power utility functions. Consider the utility function
\begin{equation}\label{eq:power-utility}
    U_{\eta}(x) = \frac{x^\eta}{\eta}, \quad \eta<1,
\end{equation}
where, by convention,  the case $\eta = 0$ corresponds to logarithmic utility. Merton's Fund Separation theorem with no consumption implies that, when
using  utility function \eqref{eq:power-utility}
with the model from Example \ref{exm-GBM}, the optimal strategy  is
\begin{equation*}
    f^{\eta}=\frac{\mu}{\sigma^2(1-\eta)}.
\end{equation*}
Comparing $f^{\eta}$ and $f^{\Ri}$ from \eqref{fRi}, we see that, as long as $\gamma>0$, the
fractional Kelly strategy with $f=f^{\Ri}$
is equivalent to optimizing power utility function with parameter $\eta = -2\gamma$.
The limits
$$
\lim_{\gamma\to 0+}f^{\Ri}(\gamma)=\lim_{\eta\to 0-}f^{\eta}=\frac{\mu}{\sigma^2}
$$
recover the original Kelly strategy, corresponding to logarithmic utility.

Theorem \ref{thm-R-simple-s} suggests  that there are many models that,
at the level of LLN and CLT, exhibit the same behavior as \eqref{R-BM}.

\begin{prop}
\label{prop-as}
Assume that
$$
\lim_{t\to \infty} r_t=\mu,\ \ \lim_{t\to \infty} v_t=\sigma^2
$$
with probability one, and, in \eqref{aux-hat}, we have
$$
\lim_{t\to \infty} \widehat{r}(t)=0,\ \ \lim_{t\to \infty} \widehat{v}(t)=0
$$
in probability. Then the conclusions of Theorem \ref{thm-R-simple-s} hold and $\vr_R(f)=f^2\sigma^2.$
\end{prop}

\begin{proof}
Convergence \eqref{ct-erg1} follows by a continuous-time version of the Toeplitz lemma. Then Proposition \ref{pro-joint-conv} implies
$$
\lim_{t\to \infty} \frac{f}{\sqrt{t}}\int_0^t \sqrt{v_s}\, dB_s \wc \mathcal{N}(0,f^2\sigma^2),
$$
completing the proof.
\end{proof}

\begin{exm}
\label{exm-logistic}
{\rm {
In \eqref{R-simple-s}, set $v_t\equiv \sigma^2$ and consider the logistic model for $r$:
$$
dr_t=r_t\left(1-\frac{r_t}{\mu}\right)\left(adt+b\,d\bar{B}_t\right),
$$
where $\bar{B}$ is a standard Brownian motion. Correlation between $B$ and $\bar{B}$ can be arbitrary, but the underlying stochastic basis must be
rich enough to support two Brownian motions. If  $r_0\in (0,\mu)$ and $ 2 a > b^2$, then the conclusion of Proposition \ref{prop-as} holds.

Indeed, with no loss of generality, we take $\mu=1$ and
verify the conditions of Proposition \ref{prop-as} using the same arguments as in
\cite[Section 4]{Logistic-Kink}.

 Consider the process
$$
Y_t=\ln\frac{r_t}{1-r_t}.
$$
By the It\^{o} formula,
$$
dY_t=\left( a+\frac{b^2}{2}\tanh(Y_t/2)\right)dt + b\,d\bar{B}_t,
$$
where $\tanh(y)=(e^y-e^{-y})/(e^y+e^{-y})$. If $\varepsilon=a-(b^2/2)>0$, then
$$
Y_t \geq Y_0+\varepsilon t + b\bar{B}_t = Y_0+t\left(\varepsilon+\frac{b}{t}\, \bar{B}_t\right).
$$
Because  $\lim_{t\to \infty} \bar{B}_t/t=0$ with probability one and
$$
0\leq 1-r_t=\frac{1}{1+e^{Y_t}} \leq e^{-Y_t},
$$
we conclude that, with probability one,
$\lim_{t\to \infty} r_t=1$ and $\int_0^{\infty} (1-r_t)\, dt <\infty$, that is, the conditions of Proposition \ref{prop-as} hold.
}}

This concludes Example \ref{exm-logistic}.
\end{exm}

Using Theorem \ref{thm-R-simple-s}, we can incorporate  standard  interest rate and stochastic volatility models into the wealth process  \eqref{wealth1-ct}.
Below  are three examples:
Vasicek model for $r$ with constant $v$, CIR model for $r$ with constant $v$, and CIR model for $v$ with constant $r$.
In all the examples,   $B$ and $\bar{B}$ are standard Brownian motions
on $(\Omega, \mathcal{F}, \{\mathcal{F}_t\}_{t\geq 0},\mathbb{P})$ such that
\begin{equation}
\label{BM-bar}
\bE [B(t)\bar{B}(s)]=\bar{\rho}\,\min(t,s),\
\end{equation}
for some $\bar{\rho}\in [-1,1]$. We will see that $\bar{\rho}$ does not affect $g_R(f)$ and $f^*$ but  does affect $\vr_R$.

\begin{exm}
\label{exm-VMr}
{\rm In \eqref{R-simple-s}, set $v_t\equiv \sigma^2$ and consider the  Vasicek model for $r$:
$$
dr_t=a\left({\mu} -r_t\right)dt+b \, d\bar{B}(t),
$$
where $a,b,\mu$ are positive numbers.  Then \eqref{CTW-f*} and \eqref{CLT-W-f} hold, with
\begin{equation}
\label{vr-VM}
\vr_R(f)=f^2\left(\frac{b^2}{a^2} +\sigma^2+  \frac{2\bar{\rho}\sigma b}{a}\right).
\end{equation}
Indeed, by direct computation,
% \begin{enumerate}
for every non-random initial condition $r_0$, $\lim_{t\to \infty} r_t\wc r^*$, and the random variable $r^*$  is normal with mean $\mu$ and
variance $b^2/(2a)$.
If $r_0$ has the same distribution as $r^*$ and is independent of $\bar{B}$, then the process $r=r_t,\ t\geq 0$, is stationary and ergodic so that
$$
\lim_{t\to \infty} \frac{1}{t}\int_0^tr_s\, ds = \bE [r^*] = {\mu}\ \ {\text{ with probability one,}}
$$
and
\begin{equation}
\label{VC-aux1}
\lim_{t\to \infty} \frac{1}{\sqrt{t}}\left( \int_0^t (r_s -\mu)ds - \frac{b\bar{B}_t}{a}\right)=0\ \ {\text{in probability.}}
\end{equation}
% \end{enumerate}
Then \eqref{CTW-f*}  holds by Theorem \ref{thm-R-simple-s}. After that, \eqref{R-simple-s-vt}, \eqref{BM-bar}, and \eqref{VC-aux1} imply
$$
\lim_{t\to \infty} \frac{1}{\sqrt{t}}\left(\ln W^f_t-t \,g_R(f)\right)
\wc \mathcal{N}\big(0,\vr_R(f)\big),\ \
$$
with $\vr_R$ given by \eqref{vr-VM}. }

 This concludes Example \ref{exm-VMr}.
 \end{exm}

\begin{exm}
\label{exm-CIRr}
{\rm In \eqref{R-simple-s}, set $v_t\equiv \sigma^2$ and consider the  CIR model \cite{CIR-Orig} for $r$:
$$
dr_t=a\left({\mu}-r_t\right)dt+b\sqrt{r_t}\,d\bar{B}(t),
$$
where $\bar{B}$ is a standard Brownian motion and $2\mu a \geq b^2$.
 Then \eqref{CTW-f*} and \eqref{CLT-W-f} hold, with
 \begin{equation}
 \label{vr-CIR}
\vr_R(f)=f^2\left(\frac{b^2\mu}{a^2} +\sigma^2+  \frac{2\bar{\rho} \sigma b\tilde{\mu}}{a}\right),
\end{equation}
where
$$
\tilde{\mu}=\frac{b\Gamma(\nu+(1/2))}{\sqrt{2a}\,\Gamma(\nu)},\ \ \nu=\frac{2a\mu}{b^2}.
$$
Indeed,
% \begin{enumerate}
for every non-random initial condition $r_0>0$, $\lim_{t\to \infty} r_t\wc r^*$, and the random variable $r^*$  has Gamma distribution with pdf
\begin{equation}
\label{GammaPDF}
p^*(x)=\frac{\alpha^\nu}{\Gamma(\nu)}x^{\nu-1}e^{-\alpha x},\ \ x\geq 0, \
\ \alpha=\frac{2a}{b^2};
\end{equation}
see \cite[Page 392]{CIR-Orig}.
If $r_0$ has the same distribution as $r^*$ and is independent of $\bar{B}$, then the process $r=r_t,\ t\geq 0$, is stationary and ergodic so that, with probability one,
 \begin{equation}
 \label{CIR-erg}
 \lim_{t\to \infty} \frac{1}{t}\int_0^tr_s\, ds = \bE [r^*] =  {\mu},
 \end{equation}
 see \cite[Theorem 2.10]{CIR-Erg}, and
 \begin{equation}
 \label{CIR-erg2}
 \lim_{t\to \infty} \frac{1}{\sqrt{t}}\left( \int_0^t (r_s -\mu)ds - \frac{b}{a}\int_0^t \sqrt{r_s}\,d\bar{B}_s\right)=0\ \ {\text{in probability.}}
 \end{equation}
% \end{enumerate}
 Then \eqref{CTW-f*}  holds by Theorem \ref{thm-R-simple-s}. After that, \eqref{R-simple-s-vt}, \eqref{BM-bar}, and \eqref{CIR-erg2} imply
$$
\lim_{t\to \infty} \frac{1}{\sqrt{t}}\left(\ln W^f_t-t \,g_R(f)\right)\wc \mathcal{N}\big(0,\vr_R(f)\big),
$$
with $\vr_R$ given by \eqref{vr-CIR}, because, by Proposition \ref{pro-joint-conv}, we have joint convergence
$$
\lim_{t\to \infty} \frac{1}{\sqrt{t}}\left( \int_0^t \sqrt{r_s}\,d\bar{B}_s, \ \ \sigma B_t\right) \wc \boldsymbol{\eta},
$$
where $\boldsymbol{\eta}$ is a bi-variate Gaussian vector with mean zero and covariance matrix
$$
\left(
\begin{array}{cc}
\mu & \bar{\rho} \sigma \tilde{\mu}\\
\bar{\rho} \sigma \tilde{\mu} & \sigma^2
\end{array}
\right),
$$
and $\tilde{\mu}=\bE\big[\sqrt{r^*}\,\big]=\frac{b\Gamma(\nu+(1/2))}{\sqrt{2a}\,\Gamma(\nu)}$.}

 This concludes Example \ref{exm-CIRr}.
 \end{exm}

\begin{exm}
\label{exm-CIRv}
{\rm In \eqref{R-simple-s}, set $r_t\equiv \mu$ and consider the  CIR model  for $v$:
\begin{equation}
\label{CIR-vol}
dv_t=\kappa\left(\sigma^2-v_t\right)dt+\beta\sqrt{v_t}\,d\bar{B}(t),\ \ 2\kappa\sigma^2\geq \beta^2.
\end{equation}
The corresponding equation \eqref{wealth1-ct}  becomes
$$
dW_t^f=fW_t^f\big(\mu dt + \sqrt{v_t}\,dB(t)\big),
$$
which, for $f=1$ and $v$ from \eqref{CIR-vol}, coincides with the Heston model originally introduced in \cite{Heston-orig}.
In this case, \eqref{CLT-W-f} and \eqref{CTW-f*} hold, with
 \begin{equation}
 \label{vr-CIRv}
 \vr_R(f)=f^2\sigma^2\left(1 +\frac{f^2\beta^2}{4\kappa^2}- 2\bar{\rho} \frac{f \beta}{\kappa}\right).
 \end{equation}

 Indeed, similar to \eqref{CIR-erg},
$$
 \lim_{t\to \infty} \frac{1}{t}\int_0^tv_s\, ds = \sigma^2,
 $$
and then  \eqref{CTW-f*}  holds by Theorem \ref{thm-R-simple-s}.

Next, similar to \eqref{CIR-erg2},
 \begin{equation}
	\label{Hest-erg2}
	\lim_{t\to \infty} \frac{1}{\sqrt{t}}\left( \int_0^t (v_s -\sigma^2)ds - \frac{\beta}{\kappa}\int_0^t \sqrt{v_s}\,d\bar{B}_s\right)=0\ \ {\text{in probability.}}
\end{equation}
 Combining  \eqref{Hest-erg2} with \eqref{R-simple-s-vt} and \eqref{BM-bar} leads to
$$
\lim_{t\to \infty} \frac{1}{\sqrt{t}}\left(\ln W^f_t-t \,g_R(f)\right)\wc
\mathcal{N}\big(0,\vr_R(f)\big),\ \
$$
with $\vr_R$ given by \eqref{vr-CIRv}, because, by Proposition \ref{pro-joint-conv}, we have joint convergence
$$
\lim_{t\to \infty} \frac{1}{\sqrt{t}}\left( \int_0^t \sqrt{v_s}\,d{B}_s, \ \ \int_0^t \sqrt{v_s}\,d\bar{B}_s \right) \wc \boldsymbol{\eta},
$$
where $\boldsymbol{\eta}$ is a bi-variate Gaussian vector with mean zero and covariance matrix
$$
\sigma^2\left(
\begin{array}{cc}
	1 & \bar{\rho} \\
	\bar{\rho}  & 1
\end{array}
\right).
$$
In particular, taking $f=2\bar{\rho}\kappa/\beta$ gives
$\vr_R=(2\bar{\rho}\kappa/\beta)^2\sigma^2(1-\bar{\rho}^2)$. As a result, if
$$
0<\frac{\beta}{\sigma^2}<\frac{2\kappa}{\beta} <\frac{\mu}{\sigma^2} < 1,
$$
then, in the limit $\bar{\rho}\to 1^{-}$, it is possible to achieve $\vr_R(f)\to 0$ and
$\SR_R(f)\to +\infty$ while satisfying the (strict) NS-NL condition  $f\in (0,1)$.}

This concludes Example \ref{exm-CIRv}. % Financial ramifications of infinite SR? Simulations showing small fluctuation? - not easy because too many approximations in the process
\end{exm}

More sophisticated constructions  are possible, for example, by combining the Vasicek or CIR model for $r$ with the CIR model for $v$.

The next example demonstrates that, for some continuous time models satisfying \eqref{R-simple-s}
and \eqref{CT-gr}, the asymptotic variance $\vr_R$ is {\em effectively} zero because
the convergence in \eqref{CLT-W-f} is much faster.

\begin{exm}
\label{exm-logistic-price}
{\rm  Consider a market where the stock price moves according to the stochastic differential equation
\begin{equation}
\label{logistic-price-sde}
dS_t=\mu S_t(M-S_t)dt + \sigma S_t dB_t,
\end{equation}
with $\mu>0$, $M>1$, and initial condition $S_0=1$. The corresponding return process $R_t$ and wealth process $W^f_t$ become
\begin{equation}
\label{LgstR}
    dR_t =  \frac{dS_t}{S_t}=\mu(M-S_t)dt + \sigma dB_t
\end{equation}
and
\begin{equation}
    \label{LgstW}
    \ln W_t^f = fR_t - \frac{f^2\sigma^2}{2}\, t.
\end{equation}
By \eqref{LgstR} and It\^{o}'s formula,
\begin{equation}
\label{LgstR1}
    R_t = \ln S_t + \frac{\sigma^2t}{2}.
\end{equation}

Assume that
\begin{equation}
\label{qqq}
2M\mu > \sigma^2.
\end{equation}
Then, by \cite[Proposition 3.3]{GieValWan-15},  the process $t\mapsto S_t$ is ergodic
and  the unique invariant distribution $S^*$
is of the form \eqref{GammaPDF} with
$$
\nu=\frac{2M\mu}{\sigma^2}-1, \ \ \alpha=\frac{2\mu}{\sigma^2}.
$$
As a result, \eqref{LgstW} and \eqref{LgstR1} imply
\begin{equation}
\label{Lgst-Conv}
    \lim_{t \to \infty}\frac{\ln W_t}{t} =\lim_{t \to \infty} \left(\frac{fR_t}{t} - \frac{f^2\sigma^2}{2} \right)=\frac{1}{2}\sigma^2f(1-f) \ {\text{ with probability one,}}
\end{equation}
so that
$$
g_R(f) = \frac{1}{2}\sigma^2f(1-f)
$$
and  $f^*=1/2$.  We see that, unlike the previous examples, $f^*$ is independent
of $\mu$ and $\sigma$.

Moreover,  \eqref{LgstW} and \eqref{LgstR1} show that the  convergence in \eqref{Lgst-Conv} is faster than  \eqref{CLT-W-f} and the corresponding limit is non-Gaussian:
\begin{equation*}
    \lim_{t \to \infty }\left({\ln W_t^f}-tg_R(f)\right) \wc f\cdot\ln S^*.
\end{equation*}
In other words,  compared to models that satisfy {\sc Conditions (C1), (C2)} and  have the
log-wealth fluctuations of order $\sqrt{t}$,
model \eqref{LgstR} leads to constant log-wealth fluctuations and thus can be considered
{\tt qualitatively} less risky.
}

This concludes Example \ref{exm-logistic-price}.
\end{exm}

There are  analogs of \eqref{TtoG-A} and \eqref{TtoG-CLT} for the process \eqref{wealth1-ct} when \eqref{R-BM} holds, that is, $R_t=\mu t + \sigma B_t,\ t\geq 0$, and
$g_R(f)=f\mu-(f^2\sigma^2/2), \ \vr_R(f)=f^2\sigma^2$.
 If  $\tw(f)=\inf\left\{ t>0: W^f_t>\ub>1\right\}$, then, according to \cite[Equation (19)]{Schrod}, the random variable $\tw(f)$ has the
 {\em inverse Gaussian distribution} with pdf
\begin{equation}
\label{fp-pdf}
p_f(t)=\left(\frac{\ln \ub}{\sqrt{\vr_R(f)}}\right)\, (2\pi t^3)^{-1/2} \exp\left(-\frac{\big(\ln \ub -t g_R(f)\big)^2}{2t\vr_R(f)}\right),\ t>0,
\end{equation}
leading to equalities
$$
\bE [\tw(f)] = \frac{\ln \ub}{g_R(f)},\ \ \
\mathrm{Var}[\tw(f)]=\frac{\vr_R(f)}{g_R^3(f)}\, \ln \ub.
$$
 Furthermore, by  \cite[Equation (1.4)]{FirstPassage},
\begin{align}
\label{TtoG-A-ct}
&\lim_{\ub\to \infty} \frac{\tw(f)}{\ln \ub }= \frac{1}{g_R(f)} \ \text{with probability one};\\
\label{TtoG-CLT-ct}
&\lim_{\ub\to \infty} \sqrt{\frac{g_R(f)}{\ln\ub}}\left(\tw(f)-\frac{\ln\ub}{g_R(f)}\right)
=\mathcal{N}\left(0, \frac{\vr_R(f)}{g_R^2(f)}\right)\
 \text{in distribution}.
\end{align}
Equality  \eqref{TtoG-A-ct} extends to some L\'{e}vy processes \cite{Gut-ct}.

To conclude our analysis of the continuous-time models, let us discuss  connections with high-frequency compounding. We start with the following observation.
\begin{prop}
\label{prop:HF-CT-1}
Given the process $R$ from \eqref{R-simple-s}, define
\begin{equation}
\label{eq:HF-CT-1-1}
P_{n,k}=e^{R_{k/n}-R_{(k-1)/n}}-1,\ W_t^{n,f}=\prod_{k=1}^{\lfloor nt \rfloor}(1+fP_{n,k}).
\end{equation}
Then
\begin{equation}
\label{eq:HF-CT-cov}
\lim_{n\to \infty} W_t^{n,f} \wcL W_t^f
\end{equation}
in $\mathcal{C}(\bR_+)$, where
\begin{equation}
\label{eq:HF-CT-lim}
W_t^f = \exp\left( fR_t+\frac{f(1-f)}{2}\int_0^t v_s\, ds\right),
\end{equation}
and the convergence is uniform in $f$ on compact subsets of $(0,1)$.
\end{prop}

\begin{proof}
Let
$$
F(x)=\ln\big(1+f(e^x-1)\big),
$$
so that
$$
F'(x)=\frac{fe^x}{1+f(e^x-1)},\ \ F''(x)=\frac{f(1-f)e^x}{(1+f(e^x-1))^2}.
$$
For $(k-1)/n<t\leq k/n$, we apply  the It\^{o} formula to the process
$$
t\mapsto F\big(R_t-R_{(k-1)/n}\big)
$$
to get
$$
\ln W_t^{n,f}= \int_0^t H^{n,f}_{t,s}\, dR_s +\int_0^t K^{n,f}_{t,s}\, ds,
$$
with
$$
H^{n,f}_{t,s}=\sum_{k=1}^{\lfloor nt \rfloor}F'(R_s-R_{(k-1)/n})\mathbf{1}_{(\frac{k-1}{n},\frac{k}{n}]}(s),\ \
K^{n,f}_{t,s}=\frac{1}{2}\sum_{k=1}^{\lfloor nt \rfloor}F''(R_s-R_{(k-1)/n})v(s)\mathbf{1}_{(\frac{k-1}{n},\frac{k}{n}]}(s).
$$
To conclude the proof, note that, for $0<s\leq t$,
\begin{align}
\label{HF-CT-p1}
&0<H^{n,f}_{t,s}\leq 1,\ \lim_{n\to \infty}H^{n,f}_{t,s}=f \ {\text{in probability}};\\
\label{HF-CT-p2}
&0\leq K^{n,f}_{t,s}\leq v_s,\ \lim_{n\to \infty}K^{n,f}_{t,s}=\frac{f(1-f)}{2}\,v_s \ {\text{in probability}}.
\end{align}
\end{proof}

Proposition \ref{prop:HF-CT-1} shows that, up to an It\^{o} correction, every continuous  wealth process is a geometric high-frequency limit of itself.
While not especially surprising, the result suggests other discrete-time approximations of the process $R$ that would lead to the same geometric high-frequency limit.
Here is one example.

\begin{thm}
\label{thm:HF-CT-1}
Given the process $R$ from \eqref{R-simple-s}, define
\begin{equation}
\label{eq:HF-CT-1-2}
P_{n,k}=\exp\left(\frac{r_{(k-1)/n}}{n}+\sqrt{v_{(k-1)/n}}\,(B_{k/n}-B_{(k-1)/n})\right)-1,\ W_t^{n,f}=\prod_{k=1}^{\lfloor nt \rfloor}(1+fP_{n,k}).
\end{equation}
If the processes $t\mapsto r_t$ and $t\mapsto v_t$ are stochastically continuous and
$$
\bE[|r_t|]<\infty, \ \ \bE [v_t]<\infty,\ \ t\geq 0,
$$
then the conclusion of Proposition \ref{prop:HF-CT-1} holds.
\end{thm}

\begin{proof}
Define the process
$$
R_{n,t} = \int_0^tr_{n,s}ds+\int_0^t\sqrt{v_{n,s}}\, dB_s,
$$
where
$$
r_{n,t}=\sum_{k\geq 1}r_{(k-1)/n}\mathbf{1}_{(\frac{k-1}{n},\frac{k}{n}]}(t) ,\
v_{n,t}=\sum_{k\geq1} v_{(k-1)/n}\mathbf{1}_{(\frac{k-1}{n},\frac{k}{n}]}(t)
$$
are the step-function approximations of $r$ and $v$.
Then the first equality in \eqref{eq:HF-CT-1-2} becomes
$$
P_{n,k}=e^{R_{n,k/n}-R_{n,(k-1)/n}}-1,
$$
making the rest of the proof identical to that of Proposition \ref{prop:HF-CT-1}. In particular, stochastic continuity of $r$ and $v$ ensures that the corresponding versions of
\eqref{HF-CT-p1} and \eqref{HF-CT-p2} hold.
\end{proof}

\section{Conclusions and Further Directions}
The aggressiveness of the full Kelly strategy is one of the most frequently discussed caveats in both theory and practice. The second-order analysis presented in this paper paves the way for further research by connecting fluctuations of the wealth process around the asymptotic long-term growth to a measure of risk. One such measure, not discussed in this work, is the Sortino ratio \cite{Sortino}, which accounts for downside fluctuations.

A fractional Kelly strategy corresponds to quantifying risk aversion.
Selecting any type of risk-aversion parameter is challenging. In the case of power utility functions,
the parameter is often estimated by analyzing individual choices in real-world scenarios, such as investment decisions or insurance purchases. Our constructions offer an alternative approach, via optimization of the asymptotic Sharpe ratio or the ridge coefficient.

%maybe write about the MacLean type analysis that can be done. Either here or after the risk-free rate>0, because that's the only case where we would like to use this approach.

A broader generalization can be achieved in various sections of this paper, including markets  with a nonzero risk-free rate,  multiple assets, and heavy-tailed returns.
In  continuous time,  heavy-tailed returns can yield more realistic and insightful results,
as well as connections with discrete-time models, such as
Example \ref{exm-CM}. One way to consider heavy-tailed returns is to add jump component to the rate process $R$; see \cite{L-P}. The corresponding second-order analysis in this setting remains an open problem.

%\bibliographystyle{amsplain}
% \bibliography{Fin}

\begin{thebibliography}{10}

\bibitem{Bellman-Kalaba}
R.~Bellman and R.~Kalaba, \emph{On the role of dynamic programming in statistical communication theory}, IRE Transactions on Information Theory \textbf{3} (1957), no.~3, 197--203.

\bibitem{Erg-1995}
R.~N. Bhattacharya and C.~Lee, \emph{Ergodicity of nonlinear first order autoregressive models}, J. Theoret. Probab. \textbf{8} (1995), no.~1, 207--219.

\bibitem{Kelly-Breiman}
L.~Breiman, \emph{Optimal gambling systems for favorable games}, Proc. 4th {B}erkeley {S}ympos. {M}ath. {S}tatist. and {P}rob., {V}ol. {I}, Univ. California Press, 1960, pp.~65--78.

\bibitem{Robbins-InfIID}
Y.~S. Chow and H.~Robbins, \emph{On sums of independent random variables with infinite moments and ``fair'' games}, Proc. Nat. Acad. Sci. U.S.A. \textbf{47} (1961), 330--335.

\bibitem{CIR-Orig}
J.~C. Cox, J.~E. Ingersoll, and S.~A. Ross, \emph{A theory of the term structure of interest rates}, Econometrica \textbf{53} (1985), no.~2, 385--407.

\bibitem{FirstPassage}
P.~Deheuvels and J.~Steinebach, \emph{On the sample path behavior of the first passage time process of a {B}rownian motion with drift}, Annales de l'I.H.P. Probabilit\'es et statistiques \textbf{26} (1990), no.~1, 145--179.

\bibitem{Kurtz-MP}
S.~N. Ethier and T.~G. Kurtz, \emph{Markov processes. {C}haracterization and convergence}, Wiley Series in Probability and Mathematical Statistics: Probability and Mathematical Statistics, John Wiley \& Sons, Inc., New York, 1986.

\bibitem{Kelly-Finkel}
M.~Finkelstein and R.~Whitley, \emph{Optimal strategies for repeated games}, Adv. in Appl. Probab. \textbf{13} (1981), no.~2, 415--428.

\bibitem{GieValWan-15}
Jean-S{\'e}bastien G., Pierre V., and Sophie W.-M., \emph{The logistic s.d.e.}, Theory of Stochastic Processes \textbf{20(36)} (2015), no.~1, 28--62.

\bibitem{Gottlieb-85}
G.~Gottlieb, \emph{An optimal betting strategy for repeated games}, Journal of Applied Probability \textbf{22} (1985), no.~4, 787–--795.

\bibitem{Gradshtein-Ryzhyk}
I.~S. Gradshteyn and I.~M. Ryzhik, \emph{Table of integrals, series, and products}, eighth ed., Elsevier/Academic Press, 2015.

\bibitem{Gut-dt-d}
A.~Gut, \emph{On the moments of some first passage times for sums of dependent random variables}, Stochastic Process. Appl. \textbf{2} (1974), 115--126.

\bibitem{Gut-ct}
A.~Gut,  \emph{On {${\rm a}.{\rm s}.$} and {$r$}-mean convergence of random processes with an application to first passage times}, Z. Wahrscheinlichkeitstheorie und Verw. Gebiete \textbf{31} (1974/75), 333--341.

\bibitem{Gut-RW}
A.~Gut,  \emph{Stopped random walks. {L}imit theorems and applications}, second ed., Springer Series in Operations Research and Financial Engineering, Springer, New York, 2009.

\bibitem{Han-Yu-Mathew-shrinakge-2019}
Y.~Han, P.~L.~H. Yu, and T.~Mathew, \emph{Shrinkage estimation of {K}elly portfolios}, Quantitative Finance \textbf{19} (2019), no.~2, 277--287.

\bibitem{hastie-09}
T.~Hastie, R.~Tibshirani, and J.H. Friedman, \emph{The elements of statistical learning: Data mining, inference, and prediction}, Springer series in statistics, Springer, 2009.

\bibitem{Heston-orig}
S.~L. Heston, \emph{A closed-form solution for options with stochastic volatility with applications to bond and currency options}, Rev. Financ. Stud. \textbf{6} (1993), no.~2, 327--343.

\bibitem{LimitTheoremsforStochasticProcesses}
J.~Jacod and A.~N. Shiryaev, \emph{Limit theorems for stochastic processes}, second ed., Springer, 2003.

\bibitem{CIR-Erg}
P.~Jin, V.~Mandrekar, B.~R\"{u}diger, and C.~Trabelsi, \emph{Positive {H}arris recurrence of the {CIR} process and its applications}, Commun. Stoch. Anal. \textbf{7} (2013), no.~3, 409--424.

\bibitem{Kelly56}
J.~L. Kelly, \emph{A new interpretation of information rate}, Bell System Technical Journal (1956), no.~35, 917--926.

\bibitem{Logistic-Kink}
P.~Kink, \emph{Some analysis of a stochastic logistic growth model}, Stoch. Anal. Appl. \textbf{36} (2018), no.~2, 240--256.

\bibitem{Klenke}
A.~Klenke, \emph{Probability theory}, third ed., Springer, 2020.

\bibitem{LSh-S1}
R.~Sh. Liptser and A.~N. Shiryaev, \emph{Statistics of random processes. {I}}, expanded ed., Applications of Mathematics (New York), vol.~5, Springer-Verlag, Berlin, 2001, General theory, Translated from the 1974 Russian original by A. B. Aries, Stochastic Modelling and Applied Probability.

\bibitem{LSh-M}
R.~Sh. Liptser and A.~N. Shiryayev, \emph{Theory of martingales}, Mathematics and its Applications (Soviet Series), vol.~49, Kluwer, 1989.

\bibitem{L-P}
S.~Lototsky and A.~Pollok, \emph{{K}elly criterion: from a simple random walk to {L}\'{e}vy processes}, SIAM J. Financial Math. \textbf{12} (2021), no.~1, 342--368.

\bibitem{MacLean-growth-security-99}
L.~C. MacLean and W.~T. Ziemba, \emph{Growth versus security tradeoffs indynamic investment analysis}, Annals of Operations Research \textbf{85} (1999), 193--225.

\bibitem{MacLean-growth-security-92}
L.~C. MacLean, W.~T. Ziemba, and G.~Blazenko, \emph{Growth versus security in dynamic investment analysis}, Management Science \textbf{38} (1992), no.~11, 1562--1585.

\bibitem{Markowitz-52}
H.~Markowitz, \emph{An optimal betting strategy for repeated games}, The Journal of Finance \textbf{7} (1952), no.~1, 77--91.

\bibitem{Merton-Fund-Separation-69}
R.~Merton, \emph{Lifetime portfolio selection under uncertainty: The continuous-time case}, The Review of Economics and Statistics \textbf{51} (1969), no.~3, 247--257.

\bibitem{ST-MC}
S.~Meyn and R.~L. Tweedie, \emph{Markov chains and stochastic stability}, second ed., Cambridge University Press, Cambridge, 2009.

\bibitem{Protter}
P.~E. Protter, \emph{Stochastic integration and differential equations}, second ed., Stochastic Modelling and Applied Probability, vol.~21, Springer, 2005.

\bibitem{Rising-Wyner-mean-shrinkage-2012}
J.~K. Rising and A.~J. Wyner, \emph{Partial {K}elly portfolios and shrinkage estimators}, {IEEE} International Symposium on Information Theory Proceedings, 2012, pp.~1618--1622.

\bibitem{Schrod}
E.~Schr\"{o}dinger, \emph{Zur {T}heorie der {F}all-und {S}teigversuche an {T}eilchen mit {B}rownscher {B}ewegung}, Physik. Zeitschr. \textbf{16} (1915), 289--295.

\bibitem{Sortino}
F.~A. Sortino and L.~N. Price, \emph{Performance measurement in a downside risk framework}, Journal of Investing \textbf{3} (1994), no.~3, 59--64.

\bibitem{Str-Var}
D.~W. Stroock and S.~R.~S. Varadhan, \emph{Multidimensional diffusion processes}, Springer-Verlag, Berlin, 2006.

\bibitem{Thorp03}
E.~O. Thorp, \emph{A perspective on quantitative finance: models for beating the market}, The Best of Wilmott 1: Incorporating the Quantitative Finance Review (P.~Wilmott, ed.), Willey, 2005, pp.~33--38.

\bibitem{Thorp06}
E.~O. Thorp, \emph{The {K}elly criterion in blackjack, sports betting, and the stock market}, The {K}elly Capital Growth Investment Criterion: {T}heory and practice (L.~C. MacLean, W.~T. Ziemba, and E.~O. Thorp, eds.), World {S}cientific Handbook in Financial Economics Series, vol.~3, World Scientific, 2011, pp.~789--832.

\bibitem{ThorpPortfolio}
E.~O. Thorp, \emph{Portfolio choice and the {K}elly criterion}, The {K}elly Capital Growth Investment Criterion: {T}heory and practice (L.~C. MacLean, W.~T. Ziemba, and E.~O. Thorp, eds.), World {S}cientific Handbook in Financial Economics Series, vol.~3, World Scientific, 2011, pp.~81--90.

\bibitem{Thorp08}
E.~O. Thorp,  \emph{Understanding the {K}elly criterion}, The {K}elly Capital Growth Investment Criterion: {T}heory and practice (L.~C. MacLean, W.~T. Ziemba, and E.~O. Thorp, eds.), World {S}cientific Handbook in Financial Economics Series, vol.~3, World Scientific, 2011, pp.~509--524.

\end{thebibliography}

%\providecommand{\bysame}{\leavevmode\hbox to3em{\hrulefill}\thinspace}
\providecommand{\MR}{\relax\ifhmode\unskip\space\fi MR }
\providecommand{\MRhref}[2]{%
  \href{http://www.ams.org/mathscinet-getitem?mr=#1}{#2}
}
\providecommand{\href}[2]{#2}

\end{document}